\documentclass[USenglish]{article}

\usepackage[utf8]{inputenc}
\usepackage[small]{dgruyter}
\usepackage{subfigure}

\theoremstyle{plain}
\newtheorem{theorem}{Theorem}[section]

\newtheorem{lemma}[theorem]{Lemma}

\theoremstyle{definition}

\newtheorem{example}[theorem]{Example}
\newtheorem{remark}[theorem]{Remark}

\newcommand{\E}{{\mathbb{E}}}

\newcommand{\R}{{\mathbb{R}}}

\renewcommand{\P}{{P}}
\newcommand{\F}{{\mathbb F}}
\renewcommand{\d}{{\text{d}}}

\newcommand{\1}{{\mathbf{1}}}
\newcommand{\cA}{{\cal A}}

\newcommand{\cC}{{\cal C}}

\newcommand{\cF}{{\cal F}}
\newcommand{\cG}{{\cal G}}
\newcommand{\cH}{{\cal H}}
\newcommand{\cJ}{{\cal J}}
\newcommand{\cK}{{\cal K}}

\newcommand{\be}{\begin{equation}}
\newcommand{\ee}{\end{equation}}
\newcommand{\bea}{\begin{eqnarray}}
\newcommand{\eea}{\end{eqnarray}}
\newcommand{\beast}{\begin{eqnarray*}}
\newcommand{\eeast}{\end{eqnarray*}}
\newcommand{\bproof}{\begin{proof}}
\newcommand{\eproof}{\end{proof}}

\newcommand{\VaR}{\text{VaR}}
\newcommand{\TCE}{\text{TCE}}
\newcommand{\EL}{\text{EL}}
\newcommand{\pe}{\boldsymbol{\pi}}
\newcommand{\bbeta}{\boldsymbol{\beta}}
\newcommand{\phe}{\boldsymbol{\varphi}}
\newcommand{\me}{\boldsymbol{\mu}}
\newcommand{\We}{\boldsymbol{W}}
\newcommand{\Se}{\boldsymbol{S}}
\newcommand{\se}{\boldsymbol{\sigma}}
\newcommand{\bu}{\boldsymbol{u}}
\newcommand{\bn}{\boldsymbol{\nu}}
\newcommand{\Rb}{\boldsymbol{R}}
\newcommand{\Loss}{L}

\begin{document}

  \articletype{~}

  \author[1]{Imke Redeker}
  \author[2]{Ralf Wunderlich}
  \runningauthor{Redeker and Wunderlich}
  \affil[1]{Mathematical Institute, Brandenburg University of Technology Cottbus - Senftenberg, Postbox 101344, 03013 Cottbus, Germany}
  \affil[2]{Mathematical Institute, Brandenburg University of Technology Cottbus - Senftenberg, Postbox 101344, 03013 Cottbus, Germany}
  \title{Portfolio optimization under dynamic risk constraints: continuous vs.~discrete time trading}
  \runningtitle{Portfolio optimization under dynamic risk constraints}
  \subtitle{}
  \abstract{We consider an investor facing a  classical portfolio problem of optimal investment in a log-Brownian stock and a fixed-interest bond, but constrained to choose portfolio and consumption strategies that reduce a dynamic shortfall risk measure.
  	For continuous- and discrete-time financial markets we investigate the loss in expected utility of intermediate consumption and terminal wealth caused by imposing a dynamic risk constraint. We derive the dynamic programming equations for the resulting stochastic optimal control problems and solve them numerically.
  	Our numerical results indicate that the loss of portfolio performance is not too large while the risk is notably reduced. We then investigate time discretization effects and find that the loss of portfolio performance resulting from imposing a risk constraint is typically bigger than the loss resulting from infrequent trading.}

\maketitle

\section{Introduction}
\label{sec:Introduction}
In a classical consumption-investment problem, an investor endowed with an initial capital consumes a certain amount of the capital and invests the remaining wealth into the financial market. The latter usually consists of one risk-free security and several risky ones. Given a fixed investment horizon, the investor's objective is to find an ``optimal'' consumption-investment strategy in order to maximize the expected utility of wealth at the terminal trading time and of intermediate consumption. This optimization problem is known in the literature as the Merton problem.

After the recent failure of large financial institutions, risk management has received a great deal of attention by traders and regulators in the past few years. Financial institutions have established internal departments that are in charge of risk assessments, and regulatory institutions started to set restrictions to limit the risk exposure of financial institutions. Various risk measures like Value at Risk (VaR) or some tail-expectation-based risk measures like Tail Conditional Expectation (TCE) or Expected Loss (EL) have been applied recently to quantify and control the risk of a portfolio. Usually the risk constraint is static, i.e., it has to hold at the terminal trading time only. In the seminal paper by \textsc{Basak} and \textsc{Shapiro} \cite{bs} the authors compute the VaR of the terminal wealth to control the risk exposure. Their findings indicate that VaR limits, when applied in a static manner, may actually increase risk. This encouraged researchers to consider a risk measure that is based on the risk neutral expectation of loss - the Limited Expected Loss (LEL). The work of \cite{bs} is extended by \textsc{Gabih} et al. \cite{gsw09} to cover the case of bounded Expected Loss.

Motivated by the Basel Committee proposals, it is a common practice in the financial industry to compute and re-evaluate risk constraints frequently using a time window (e.g. some days or weeks) over which the trading strategies are assumed to be held constant, cf. \textsc{Jorion} \cite{jor}. Therefore \textsc{Cuoco} et al. \cite{cuo} apply a risk constraint dynamically to control the risk exposure. The authors found that VaR and TCE constraints, when applied in a dynamic fashion, reduce the investment in the risky asset. \textsc{Putsch\"ogl} and \textsc{Sass}  \cite{PutschoeglSass2008} study the maximization of expected utility of terminal wealth under dynamic risk constraints in a  complete market model with partial information on the drift by  using  the  martingale method. The problem of intertemporal consumption subject to a dynamic VaR constraint is studied by \textsc{Yiu} \cite{yiu}. \textsc{Pirvu} \cite{pir} and \textsc{Leippold} et al. \cite{leippold2006} study an optimal consumption-investment problem with a dynamic VaR constraint imposed on the strategy. \textsc{Akume} \cite{aku} and \textsc{Akume} et al. \cite{aku0} consider a similar problem with a dynamic TCE constraint instead of a VaR. Their results indicate that imposing a dynamic risk constraint is a suitable method to reduce the risk of portfolios. For recent work on portfolio optimization under dynamic risk constraints we refer to \textsc{Moreno-Bromberg} et al. \cite{mpr}. The authors study an optimal investment problem for a general class of risk measures and the optimal trading strategy is characterized by a quadratic backward stochastic differential equation.

For the purpose of dynamic risk measurement, the usual assumption for a tractable way of calculating  risk measures is that the investment and consumption strategies are kept unchanged over a given time horizon. In the classical consumption-investment problem, this assumption is typically not fulfilled. An investor who follows the optimal strategy is continuously adjusting the asset holdings and the consumption rate. Thus, the risk of a portfolio is only approximately calculated.

Since continuous-time trading is impossible in practice and investors  want to avoid transaction costs  due to excessive trading we also consider a more realistic scenario where investors are only allowed to change the asset holdings and consumption rate at discrete points in time.
In \textsc{Rogers} \cite{rog} and \textsc{Bäuerle} et al. \cite{buv} such a discrete-time investor is called  a relaxed investor. It is shown that the losses due to  discretization are surprisingly small.  An advantage of the discrete-time trading is that dynamic risk measures can be computed directly without requiring constant strategies.

In this paper, we start with considering  a continuous-time investor facing a  classical con\-sump\-tion-in\-vest\-ment problem, but constrained to choose investment and consumption strategies that reduce a corresponding shortfall risk.
The investor's aim is to maximize the expected utility of the intermediate consumption and wealth at a terminal trading time $T > 0$. The portfolio risk over short time intervals is measured in terms of  VaR, TCE and EL. The risk measure is  dynamically re-evaluated using available conditioning information and imposed on the strategy as a risk constraint. We apply dynamic
programming techniques and combine the resulting Hamilton-Jacobi-Bellman equation with the method of Lagrange multipliers to derive optimal strategies under this constraint. An approximate
solution to the constrained portfolio problem is obtained by using a policy improvement algorithm. One advantage of our method  compared to \cite{mpr} is that the consumption of wealth can easily be included.

Then, we consider  a discrete-time investor who only changes portfolio and consumption choices at time points that are multiples of $\Delta >0$. Within the framework of a discrete-time investment, we consider a discretized standard Black-Scholes market with one risk-free and several risky assets. A discrete-time investor is, like a continuous-time investor, constrained to limit the risk exposure.  The optimization problem is solved by using the theory of Markov Decision Problems leading to a backward recursion algorithm.

Our numerical results indicate that the cost of the risk constraint in terms of the expected utility of intermediate consumption and terminal wealth is not too large while the risk can be controlled and considerably reduced. The consideration of both, discrete-time and continuous-time trading, allows us to perform numerical experiments for studying the  losses due to time-discretization.  Similar to Rogers \cite{rog} who investigates the Merton problem, we find that  for the portfolio problem under dynamic risk constraints the effects of time-discretization are also of small magnitude.

The paper is organized as follows. In Section \ref{sec2} we consider the continuous-time optimization problem. First, we describe the financial market model and introduce dynamic risk constraints afterwards. Then, we investigate the optimization problem under dynamic risk constraints. In Section \ref{sec3}, we consider the discrete-time optimization problem. The financial market of Section \ref{sec2} is discretized and the resulting dynamic risk constraints are described. We subsequently study the discrete-time optimization problem under dynamic risk constraints. Section \ref{sec4} provides some numerical results of the continuous-time and discrete-time optimization problem and both are compared. The Appendix contains proofs omitted from the main text.

\section{Continuous-time optimization}\label{sec2}
\subsection{Model}
We consider a continuous-time stochastic financial market with finite trading horizon $0<T<\infty$. The possible actions of an investor who is endowed with an initial capital $x_0 > 0$ are to invest in the financial market and/or to consume (parts of) the wealth. The investment opportunities are represented by one risk-free and $d$ risky securities. The price of the risk-free security at time $t$ is denoted by $S_t^0$ and the $d$-dimensional price process of the risky securities is denoted by $\Se=(\Se_t)_{t\in[0,T]}$. Uncertainty is modeled by a filtered probability space $(\Omega,\cF,\F,\P)$, where $\F =(\cF_t)_{t\in [0,T]}$ is the natural filtration generated by an $m$-dimensional Brownian motion $\We=(\We_t)_{t\in[0,T]}$, $m\ge d$, augmented by all the $\P$-null sets of $\Omega$.\\
The risk-free security (the ``bond'') behaves like a bank account earning a continuously compounded interest rate $r\ge0$, i.e., its price is given by
\begin{align*}
S_t^0 = e^{rt}.
\end{align*}
The remaining $d$ securities (the ``stocks'') are risky and evolve according to the following stochastic differential equations
\begin{align*}
\d S_t^i &= S_t^i\Big(\mu^i \d t + \sum_{j=1}^m \sigma^{ij}\d W_t^j\Big), \quad i=1,\ldots,d,\\
S_0^i &= s^i,
\end{align*}
where $\me=(\mu^i)_{i=1,\ldots,d} \in \R^d$ is the mean rate of return and $\se = (\sigma^{ij})^{j=1,\ldots,m}_{i=1,\ldots,d} \in \R^{d\times m}$ the matrix-valued volatility which is assumed to consist of linearly independent rows. Thus, we work on a standard Black-Scholes-Merton model. The bond and the stocks can be traded continuously and we allow to sell the stock and the bond short. An investment strategy $\pe=(\pe_t)_{t\in[0,T]}$ is an $\F$-adapted, real-valued, $d$-dimensional stochastic process satisfying
\begin{align*}
\E\Big[\int\nolimits_0^T \lVert\pe_t \rVert^2\d t\Big] < \infty,
\end{align*}
where $\lVert\cdot \rVert$ denotes the standard Euclidean norm in $\R^d$. 
The $d$ coordinates of $\pe_t$ represent the proportions of the current wealth invested in each of the $d$ stocks at time $t$.

A negative proportion $\pi_t^i$ corresponds to selling stock $i$ short. The proportion of wealth invested in the bond at time $t$ is given by $1-\sum_{i=1}^d\pi_t^i$. If this quantity is negative we sell the bond short, i.e., we are borrowing (at interest rate $r$). The investor is also allowed to withdraw funds for consumption. The consumption rate process (for brevity: consumption process) is denoted by $C=(C_t)_{t\in [0,T]}$. It is an adapted, $[0,\infty)$-valued stochastic process satisfying
\begin{align*}
\E\Big[\int\nolimits_0^T C_t \,\d t\Big] <\infty.
\end{align*}
Given an investment and consumption strategy the associated wealth process is well defined and satisfies the stochastic differential equation
\begin{align}\begin{split}\label{eq:wealth}
\d X_t &= X_t \Big(1-\sum_{i=1}^d\pi_t^i\Big)r \,\d t + X_t  \pe_t'\me\, \d t - C_t \d t  + X_t \pe_t'\se \d \We_t\\
&= X_t \left(r + \pe_t'(\me -\1 r)- c_t\right)\d t + X_t \pe_t'\se \d \We_t,
\end{split}
\end{align}
where $\pe_t'$ denotes the transpose of $\pe_t$, $\mathbf{1} := (1,\ldots,1)'$ is the $d$-dimensional vector with unit components and $c_t = C_t/X_t$ is the consumption rate relative to the wealth. Together with the initial condition $X_0 = x_0$, the wealth equation \eqref{eq:wealth} admits a unique strong solution given by
\begin{align}\label{eq:wealth_sol}
X_t = x_0 \exp\bigg\{\int_0^t \Big(r + \pe_s'(\me -\1 r)-c_s-\frac{1}{2}\lVert \pe_s'\se\rVert^2 \Big) \d s + \int_0^t \pe_s'\se \d \We_s \bigg\}.
\end{align}
Note that \eqref{eq:wealth_sol} implies
\begin{align}\label{eq:wealth_sol1}
X_{t+\Delta} = X_t \exp\bigg\{\!\!\int_t^{t+\Delta} \!\!\big(r + \pe_s'(\me -\1 r) \!-c_s\! -\frac{1}{2}\lVert \pe_s'\se\rVert^2 \big) \d s +\!\! \int_t^{t+\Delta} \!\pe_s'\se \d \We_s \!\bigg\}
\end{align}
for any $\Delta >0$. A pair $(\pe_t,c_t)_{t\in[0,T]}$
is called a portfolio-proportion process. We only consider those portfolio-proportion processes $\bu = (\pe,c)$ which achieve a positive wealth over the whole trading period $[0,T]$ and which are of Markov type, i.e., $\bu_t = \widetilde  u(t,X_t)$ for all $t\in[0,T]$ and some measurable function $\widetilde  u \colon [0,T]\times \R_+ \to \cK $ where $\cK:=\R^d \times [0,\infty)$. Such strategies are called admissible and the set of admissible strategies is denoted by $\cA_0$, thus
\begin{align*}
\cA_0 := \Big\{& (\bu_t)_{t\in[0,T]} \,\Big| \,\bu \text{ is } \F\text{-adapted}, \bu_t=(\pe_t,c_t) \in \cK, \bu_t = \widetilde  u(t,X_t),  \\
& X_t > 0\text{ for all }t\in[0,T]
\text{ and } \E\Big[\int\nolimits_0^T \big(\lVert\pe_t \rVert^2 + c_t\big)\d t\Big] < \infty\Big\}.
\end{align*}
We write in the following $X^{\bu}$ instead of $X$ to emphasize that the wealth is controlled by the portfolio-proportion process $\bu=(\pe,c)$.
\subsection{Dynamic risk constraints}
In this section we introduce how the risk of a given portfolio-proportion process can be quantified. In \textsc{Artzner} et al. \cite{art} risk is 
defined by the random future value of the portfolio wealth. In order to relate this definition of risk to the investor's loss we use the concept of benchmarks as in \cite{aku} and \cite{aku0}.\\
Given the current time $t \in[0,T]$ a benchmark $Y_t$ is prescribed and compared to the future portfolio value $X^{\bu}_{t+\Delta}$ at time $t+\Delta$, where $\Delta>0$ is the length of the risk measurement horizon $[t,t+\Delta]$. Then a shortfall is described by the random event $\{X^{\bu}_{t+\Delta} < Y_t \}$ and $L_t:=Y_t-X^{\bu}_{t+\Delta} $ is the corresponding investor's loss. The benchmark $Y_t$ is chosen as a function of time $t$ and wealth $X^{\bu}_t$ for all $t\in [0,T]$, i.e., $Y_t = \widetilde  f(t,X^{\bu}_t)$ for some measurable function $\widetilde  f\colon [0,T]\times \R_+\to [0,\infty)$. Typical benchmarks are presented in the following example.
\begin{example}\label{benchmark_examples}
	The benchmark may be chosen as
	\begin{itemize}
		\item a constant, i.e., $Y_t \equiv y$ for some $y\geq0$,
		\item a deterministic function, i.e., $Y_t = y(t)$,
		\item a fraction $p>0$ of the current wealth, i.e., $Y_t = p X^{\bu}_t$,
		\item the conditional expected wealth, i.e., $Y_t = \E[X^{\bu}_{t+\Delta}|\cF_t]$.
	\end{itemize}	
\end{example}
Next, we make precise how the risk of a given portfolio-proportion process is measured. Let
	\begin{align*}
	\cA :=\big\{(\bu_t)_{t\in[0,T]} \in \cA_0 \, |\,  
	\mathbb{E}[\vert L_t \vert]<\infty \text{ for all } t\in [0,T]\big\}
	\end{align*}
	and 
	\begin{align*}
	\mathcal{N}_t := \left\{L_t=Y_t - X^{\bu}_{t+\Delta}\big| \bu \in \cA \right\}.
	\end{align*}
Then the family $(\xi_t)_{t\in[0,T]}$ of maps $\xi_t$ with
\begin{align*}
\xi_t \colon \mathcal{N}_t  \to \mathcal{L}^1(\Omega,\cF_t,\P)
\end{align*}
is called a dynamic risk measure. The following dynamic risk measures are frequently used in the literature.
\begin{example}\label{ex:risk_measures}
	\ 
	\begin{itemize}
		\item  Given a time $t$ and a probability level $\alpha \in (0,1)$, the \textbf{Dynamic Value at Risk} denoted by $\VaR_t^{\alpha}$ is the loss over $[t,t+\Delta]$ that is exceeded only with the (small) conditional probability $\alpha$, thus
		\begin{align*}
		\xi_t(\Loss_t)=\VaR_t^\alpha(\Loss_t) := \inf\left\{l \in \R \,|\, \P(\Loss_t >l|\cF_t) \leq \alpha \right\}.
		\end{align*}
		\item Given a time $t$ and a probability level $\alpha \in (0,1)$, the \textbf{Dynamic Tail Conditional Expectation} denoted by $\TCE_t^\alpha$ is the conditional expected value of the loss exceeding $\VaR_t^\alpha$, thus
		\begin{align*}
		\xi_t(\Loss_t)=\TCE_t^\alpha(\Loss_t) := \E_t\left[\Loss_t|\Loss_t > \VaR_t^\alpha(\Loss_t)\right],
		\end{align*}
		where $\E_t[\cdot]$ denotes the conditional expectation given the information known up to time $t$.
		\item Given a time $t$, the \textbf{Dynamic Expected Loss} denoted by $\EL_t$ is the conditional expected value of ``positive'' losses, thus
		\begin{align*}
		\xi_t(\Loss_t)=\EL_t(\Loss_t) := \E_t\left[\Loss_t^+\right],
		\end{align*}
		where $x^+ = \max(x,0)$.
	\end{itemize}	
\end{example}
In the following we only consider those risk measures that can be written as $\xi_t(\Loss_t) = \widetilde  \xi(t,X^{\bu}_t,\pe_t,c_t)$ for all $t\in [0,T]$ and some measurable function $\widetilde  \xi \colon [0,T ]\times\R_+\times \cK \to \R$. The risk measures presented in Example \ref{ex:risk_measures} belong to this class of risk measures. This becomes obvious if we recall that loss is defined by $\Loss_t = Y_t - X^{\bu}_{t+\Delta}$ and that the benchmark is of the form $Y_t=\widetilde f(t,X^{\bu}_t)$. Thus, we have to know the conditional distribution of $X^{\bu}_{t+\Delta }$ given $X^{\bu}_{t}$ at any time $t$ to explicitly compute the risk measures above. From Equation \eqref{eq:wealth_sol} and Equation \eqref{eq:wealth_sol1} it is easily seen that the distribution of the investor's wealth at a future date depends on the portfolio-proportion process $\bu=(\pe,c)$. For the purposes of risk measurement, it is common practice to approximate this distribution (for elaborations see \cite{cuo}). Let us consider the random variable
\begin{align*}
\mathcal{X}=\mathcal{X}(x,\bar \pe, \bar c)
=x \exp\Big\{&\Big(r + \bar\pe'(\me -\1 r)-\bar c-\frac{\lVert \bar \pe'\se\rVert^2}{2} \Big) \Delta \\
&+  \bar \pe'\se \left(\We_{t+\Delta}-\We_t\right) \Big\},
\end{align*}
where $\Delta >0$, $x>0$ and $(\bar \pe,\bar c) \in \cK$ are given. It is easily seen that $\mathcal{X}$ is log-normally distributed. More precisely, the law of $\mathcal{X}(x,\bar \pe, \bar c)$ is the one of $xe^Z$, where $Z$ is a normally distributed random variable with mean $\left(r + \bar\pe'(\me -\1 r)-\bar c-\frac{1}{2}\lVert \bar \pe'\se\rVert^2\right)\Delta$ and variance $\lVert \bar \pe'\se\rVert^2 \Delta$. We immediately obtain from Equation \eqref{eq:wealth_sol1} that, given a portfolio-proportion process $\bu=(\pe,c)$ and the associated portfolio wealth $X^{\bu}_t$ at time $t$, the random variable $\mathcal{X}(X^{\bu}_t,\pe_t,c_t)$ is  the value of the portfolio wealth at time $t+\Delta$, if the portfolio-proportion process is kept constant at $(\bar\pe,\bar c)=(\pe_t,c_t)$ between time $t$ and $t+\Delta$. Then, $X^{\bu}_{t+\Delta}$ is - conditionally on $\cF_t$ -  distributed as $\mathcal{X}(X^{\bu}_t,\pe_t,c_t)$. Using this approximation we obtain the following formulas for the risk measures introduced in the examples above.
\begin{lemma}\label{lem:risk_measures}\leavevmode
	\begin{enumerate}
		\item The Dynamic Value at Risk at time $t$ can be written as $	\VaR_t^\alpha\left(\Loss_{t}\right)=\widetilde \xi(t,X_t^{\bu},\pe_t,c_t)$, where
		\begin{align*}
		\widetilde  \xi(t,x,\bar \pe,\bar c)= \widetilde  f(t,x) - x\exp\Big[&\Big(\bar\pe(\me-\1 r)+r-\bar c-\frac{\lVert\bar\pe\se\rVert^2}{2}\Big)\Delta \\	&+\Phi^{-1}(\alpha)\lVert\bar\pe\se\rVert \sqrt{\Delta }\Big].
		\end{align*}
		Here $\Phi(\cdot)$ and $\Phi^{-1}(\cdot)$ denote the normal distribution and the inverse distribution
		functions, respectively.
		\item The Dynamic Tail Conditional Expectation at time $t$ can be written as $\TCE_t^\alpha\left(\Loss_{t}\right) = \widetilde \xi(t,X_t^{\bu},\pe_t,c_t)$, where
		\begin{align*}
		\widetilde  \xi(t,x,\bar \pe,\bar c)=  \widetilde  f(t,x) - \frac{x}{\alpha}\big[&\exp\left\{\left(\bar\pe(\me-\1 r)+r-\bar c\right)\Delta \right\}\\
		&\Phi\big(\Phi^{-1}(\alpha)-\lVert\bar\pe\se\rVert\sqrt{\Delta}\big)\big].
		\end{align*}
		\item The Dynamic Expected Loss at time $t$ can be written as $\EL_t\left(\Loss_{t}\right) =\widetilde \xi(t,X_t^{\bu},\pe_t,c_t)$, where
		\begin{align*}
		&\widetilde  \xi(t,x,\bar \pe,\bar c)=\widetilde  f(t,x)\Phi(d_1)-x\exp{\left\{\left(\bar\pe(\me-\1 r)+r-\bar c\right)\Delta \right\}}\Phi(d_2)
		\end{align*}
		and
		\begin{align*}
		d_{1/2} = \frac{1}{\lVert\bar\pe\se\rVert\sqrt{\Delta}}\bigg[ \ln\bigg(\frac{\widetilde  f(t,x)}{x}\bigg)-\Big(\bar\pe(\me-\1 r)+r-\bar c\mp\frac{\lVert\bar\pe\se\rVert^2}{2}\Big)\Delta \bigg].
		\end{align*}
	\end{enumerate}
\end{lemma}
\begin{proof}
	The first and the second claim are proved in \cite[pp. 144-145]{aku} and \cite[pp. 145-147]{aku}, respectively. The proof of the third claim is presented  in Appendix \ref{appA}.
\end{proof}
A risk constraint is imposed on the portfolio-proportion process by requiring that $\bu_t = (\pe_t,c_t)$ takes values in the set $\cK^R(t,X_t^{\bu})$ at any time $t$, where $\cK^R(t,X_t^{\bu})$ is defined by
\begin{align*}
\cK^R(t,x)=\big\{(\bar\pe,\bar c) \in \R^d \times [0,\infty) \,\,\big| \,\,\widetilde  \xi(t,x,\bar \pe,\bar c) \leq \widetilde  \varepsilon(t,x)\big\}, \qquad x>0.
\end{align*}
Here, $\widetilde  \varepsilon \colon [0,T]\times \R_+ \to [0,\infty)$ is a measurable function which represents the bound on the risk constraint and may depend on time and wealth. Then the set of admissible strategies which continuously satisfy the imposed risk constraint reads as
\begin{align*}
\cA^R :=\big\{(\bu_t)_{t\in[0,T]} \in \cA \, |\,  \bu_t \in \cK^R(t,X_t^{\bu}) \text{ for all } t\in [0,T]\big\}.
\end{align*}
We assume that the benchmark $Y_t=\widetilde f(t,X^{\bu}_t)$ and the bound $\widetilde  \varepsilon(t,X^{\bu}_t)$ in the definition of the risk constraint are chosen such that $\cA^R \neq \varnothing$ is satisfied, i.e., there exist admissible strategies. Such problems are examined in \cite{gsw09}, where the shortfall risk is measured in terms of the expected loss and applied in a static manner.
\subsection{Optimization under risk constraints}\label{sec2_3}
Given a finite trading horizon $T$ we consider the problem of an investor who starts with a positive endowment $X_0=x_0$. The investor derives utility from intermediate consumption and from terminal wealth while a risk constraint is imposed that has to be satisfied. The investor's performance criterion is the expected value of the utility of intertemporal consumption and terminal wealth. Thus, the objective is to maximize
\begin{align*}
\E_{0,x_0}\Big[\int\nolimits_0^T U_1(C_s)\d s + U_2\left(X^{\bu}_T\right)\Big]
\end{align*}
over all portfolio-proportion processes $\bu =(\pe,c) \in \cA^R$. Note that $C_s=c_sX^{\bu}_s$. Here $U_1, U_2 \colon [0,\infty)\to\R \cup \{-\infty\}$  denote (time-independent) utility functions (i.e., $U_1$ and $U_2$ are strictly increasing, strictly concave and twice continuously differentiable on $(0,\infty)$). The term $\E_{t,x}[\cdot]$ denotes the conditional expectation given the information known up to time $t$ and $X_t = x$.
	\begin{remark}	
		We consider time-independent utility functions for the sake of notational simplicity but time-dependent utility functions can be treated in the same way. A typical example of a time-dependent utility function $U \colon [0,T]\times [0,\infty) \to\R\cup \{-\infty\}$ is 
		\begin{align*}
		U(t,x)=e^{-\rho t} U_1(x), \qquad \rho > 0.
		\end{align*}	
		Thus, discounting might be included.
	\end{remark}
We tackle the optimization problem by using dynamic programming techniques and embed it into a family of optimization problems. Given a portfolio-proportion process $\bu=(\bu_t)_{t\in[0,T]} \in \cA^R$ we define the reward function $\cJ(t,x,\bu)$ for all $(t,x) \in [0,T] \times \R_+$ by
\begin{align*}
\cJ(t,x,\bu)= \E_{t,x}\Big[\int\nolimits_t^T U_1(c_sX^{\bu}_s)\d s + U_2(X_T^{\bu})\Big].
\end{align*}
The objective is to maximize the reward function over all admissible processes which continuously satisfy the imposed risk constraint. We define the value function for all $(t,x) \in [0,T] \times \R_+$ by
\begin{align}\label{SCP}
V(t,x)=\sup_{\bu\in \cA^R}\cJ(t,x,\bu).
\end{align}
A portfolio-proportion process $\bu^{\ast}$ with $V(t,x) = \cJ(t,x,\bu^{\ast})$ is called optimal.
The Hamilton-Jacobi-Bellman (HJB) equation for $V$ is derived by applying the dynamic programming principle, cf. \textsc{Pham} \cite[Section 3.3, Theorem 3.3.1]{pha}, which yields
\begin{align}\label{HJB}
\frac{\partial}{\partial t}V(t,x) +\sup_{\bar \bu=(\bar\pe,\bar c) \in \cK^R(t,x)} \left\{U_1(\bar c x)+\cH^{\bar \bu} V(t,x)\right\}=0 
\end{align}
for $(t,x) \in [0,T)\times \R_+$ with terminal condition $V(T,x) =  U_2(x)$ for $x \in \R_+$. Here, the operator $\cH^{\bar \bu}$ acting on $\cC^{1,2}([0,T]\times \R_+)$ is defined by
\begin{align*}
\cH^{\bar \bu}= \left[x \left(\bar{\pe}'(\me-\1 r)+r-\bar c\right)\right]\frac{\partial }{\partial x}+\frac{1}{2}x^2 \bar\pe'\se \se'\bar\pe \frac{\partial^2 }{\partial x^2}.
\end{align*}
Note that $\cH^{\bar \bu}$ is the generator of the controlled wealth process when the portfolio-proportion process takes the value $\bar \bu =(\bar\pe,\bar c)$. Unlike many stochastic control problems, our formulation poses a state-dependent set of admissible controls. This difficulty can be handled by embedding the state-dependent set  $\cK^R(t,X_t^{\bu})$ into a compact set $\bar \cK$, cf. \textsc{Shardin} and \textsc{Wunderlich} \cite{sha}.
Then, using the techniques presented in \textsc{Fleming} and \textsc{Soner} \cite[Chapter III, Theorem 8.1]{fle} one can show that if there exists a classical solution $\widetilde V$ of the HJB equation \eqref{HJB} then $\widetilde V$ coincides with the value function $V$ of the control problem \eqref{SCP}. Furthermore, if there exists a measurable function $\widetilde  u^{\ast} \colon [0, T) \times \R_+ \to \cK$ satisfying  $\widetilde  u^{\ast}(t,x) \in \cK^R(t,x)$ for all $(t,x) \in [0,T)\times \R_+$ and such that for every $(t,x) \in [0,T)\times \R_+$ the value of $\widetilde  u^{\ast}$ at $(t,x)$ is the unique maximizer of the problem
\begin{align*}
\sup_{\bar \bu=(\bar\pe,\bar c) \in \cK^R(t,x)} \left\{U_1(\bar c x)+\cH^{\bar \bu} V(t,x)\right\}
\end{align*}
then one can show that the optimal portfolio-proportion process is given by $\bu^{\ast}_t = \widetilde  u^{\ast}(t,X^{u^{\ast}}_t)$ and satisfies $\bu^{\ast} \in \cA^R$.\\\\
For the case of power utility, i.e.,
	\begin{align*}
	U_1(x)=U_2(x)=\frac{x^{1-\gamma}}{1-\gamma}, \quad 0<\gamma <1, 
	\end{align*}
	there is a closed-form expression for the optimal portfolio-proportion process when no risk constraint is imposed, see \textsc{Korn} \cite{kor}. We call this process the Merton portfolio-proportion process and it is given by
	\begin{align}\label{Merton_strategy}
	\boldsymbol{\pi}^{M}(t,x) &=\boldsymbol{\pi}^{M}= (\sigma \sigma')^{-1}(\me-\1 r)\frac{1}{\gamma}\quad \text{and} \\
	c^{M}(t,x)&=c^{M}(t)=\big({\tau}^{-1}+(1-{\tau}^{-1})e^{-{\tau}(T-t)}\big)^{-1}, 	
	\end{align}
	where
	\begin{align*}
	{\tau} = \Big[-(1-\gamma) \Big(\frac{(\me-\1 r)'(\sigma \sigma')^{-1}(\me-\1 r)}{2\gamma}+r\Big)\Big]\frac{1}{\gamma}.
	\end{align*}
	An investor using the Merton portfolio-proportion process is called a Merton investor and the associated value function is denoted by $V^M(t,x)$. We write $X^{M}$ when the wealth is controlled by the Merton portfolio-proportion process $\bu^M:=(\pe^M,c^M)$.

\section{Discrete-time optimization}\label{sec3}
\subsection{Model}
Here we suppose that the trading interval $[0,T]$ is divided into $N$ periods of length $\Delta$ and trading only takes place at the beginning of each of the $N$ periods. The trading times are denoted by $t_n := n \Delta$, $n=0,\ldots,N-1$, and for the time horizon $T$ we write  $t_N:= T= N \Delta$.
Uncertainty is modeled by the filtered probability space $(\Omega,\cG,\mathbb{G},\P)$ where the filtration $\mathbb{G}=(\cG_{t_n})_{n=0,\ldots,N}$ is generated by the (discretized) Brownian motion $(W_{t_n})_{n=0,\ldots,N}$ and augmented by all the $\P$-null sets of $\Omega$. In what follows we will consider an $N$-period financial market which results from a discretization of the Black-Scholes-Merton model considered in the continuous-time case and consists of one risk-free and $d$ risky securities. Thus, the price process of the bond is given by $S^0_{t_0} \equiv 1$ and
\begin{align*}
S^0_{t_{n+1}} = S^0_{t_n} e^{r \Delta}, \quad n=0,\ldots,N-1.
\end{align*}
Recall, $r \ge 0$ denotes the continuously compounded interest rate. The price process of each of the $d$ risky securities is given by $S_{t_0}^k=s^k$ with $k=1,\ldots,d$ and
\begin{align*}
S_{t_{n+1}}^k = S_{t_n}^k \widetilde R_{t_{n+1}}^k, \quad k=1,\ldots,d, \quad n=0,\ldots,N-1,
\end{align*}
where $\widetilde R^k_{t_{n+1}}$ is defined by
\begin{align*}
\widetilde R^k_{t_{n+1}} = \exp\Big\{\Big(\mu^k - \frac{1}{2}\sum_{j=1}^m \left(\sigma^{kj}\right)^2\Big)\Delta + \sum_{j=1}^m \sigma^{kj}\left(W^j_{t_{n+1}}-W^j_{t_n}\right)\Big\}
\end{align*}
for $k=1,\ldots,d$. The random variables $\widetilde R^k_{t_{n+1}}$ are log-normally distributed and represent the relative price change of the risky securities in the time interval $[t_n,t_{n+1})$.\\ 
Analogously to the continuous-time case the investor starts with an initial wealth $x_0 >0$ and is allowed to invest in the financial market and to consume the wealth. In contrast to the continuous-time case the investor can only adjust at the beginning of each of the $N$ periods the amount of wealth invested into the financial market and the amount of wealth consumed. The amount invested in the $d$ risky securities is denoted by $\phe = (\phe_{t_n})_{n\in \{0,\ldots,N-1\}}$ and the amount which is consumed by $\eta= (\eta_{t_n})_{n\in \{0,\ldots,N-1\}}$. Given an investment-consumption strategy $\bn=(\phe,\eta)$ the associated wealth process evolves as follows
\begin{align*}
X_{t_{n+1}} &= e^{r\Delta}(X_{t_n}-\eta_{t_n}-\phe_{t_n}'\1)+\phe_{t_n}'\cdot \widetilde \Rb_{t_{n+1}}\\
&= e^{r\Delta}(X_{t_n}-\eta_{t_n}+\phe_{t_n}'\cdot \Rb_{t_{n+1}}),
\end{align*}
where $\Rb_{t_{n+1}} = (R^1_{t_{n+1}},\ldots,R^d_{t_{n+1}})$ denotes the relative discounted return process and is defined by
\begin{align*}
R^k_{t_{n+1}}=e^{-r\Delta}{\widetilde R^k_{t_{n+1}}}-1 \quad k=1,\ldots,d.
\end{align*}
The investment strategy $(\phe_{t_n})$ and the consumption strategy $(\eta_{t_n})$ are assumed to be $\mathbb{G}$-adapted. Moreover, we restrict to strategies which achieve a positive wealth for all $N$ periods. Thus, in contrast to the continuous-time case, for an investor who only trades at discrete points in time it is not admissible to sell stocks short or to take out a loan to buy stocks because in a time interval of length $\Delta$ the stock price could move unfavorably for the investor, and the wealth would become negative. Therefore, at any trading time $t_n$, $n\in \{0,\ldots,N-1\}$, the investor must choose the amount $\phe_{t_n}$ of current wealth $X_{t_n}$ that is invested in the stocks and the amount $\eta_{t_n}$ of current wealth $X_{t_n}$ that is consumed in such a way that $\bn_{t_n}=( \phe_{t_n},  \eta_{t_n}) \in \cK_D (X_{t_n})$ holds, where 
\begin{align*}
\cK_D (x)= \left\{\bar\bn=(\bar \phe,\bar \eta) \, | \, 0\leq \bar \eta \leq x \text{ and } 0\leq \bar \phe' \1 \leq x-\bar \eta\right\},\qquad x>0.
\end{align*}
Furthermore, we only consider strategies which are of Markov type, i.e., $\bn_{t_n} = \widetilde  \nu(t_n,X_{t_n})$ for all $n=0,\ldots,N-1$ and some measurable function $\widetilde  \nu \colon [0,T]\times\R_+ \to [0,\infty)^{d+1}$.
Such strategies are called admissible and the set of admissible strategies is denoted by $\cA^0_{\Delta}$, thus
\begin{align*}
\cA^0_{\Delta}:= \big\{(\bn_{t_n})_{n\in\{0,\ldots,N-1\}} \,\Big| \,\bn \text{ is } \mathbb{G}\text{-adapted}, \bn_{t_n}=(\phe_{t_n},\eta_{t_n}) \in \cK_D (X_{t_n}),\\
\bn_{t_n} = \widetilde  \nu(t_n,X_{t_n}) \text{ and } X_{t_n} > 0 \text{ for all }n\in\{0,\ldots,N-1\} \big\}.
\end{align*}
It is shown in \textsc{Hinderer} \cite[Theorem 18.4]{hin} that  more general strategies which depend on the complete history of the process instead of being Markovian ones, do not increase the value of the maximization problem that we will consider later on. We write in the following $X^{\bn}$ instead of $X$ to emphasize that the wealth is controlled by the investment-consumption strategy $\bn=(\phe,\eta)$.
\subsection{Dynamic risk constraints}
Analogously to the continuous-time case the loss over the period $[t_{n},t_{n+1})$ is defined by $\Loss_{t_n} := Y_{t_n}-X_{t_{n+1}}^{\bn}$ with $Y_{t_n}$ being a $\cG_{t_n}$-measurable benchmark (see Example \ref{benchmark_examples}) prescribed at time $t_n$. Let $\mathcal{N}^{\Delta}_{t_n} := \left\{{L_{t_n}}\big| \bn \in \cA_{\Delta} \right\}$, where
	\begin{align*}
	\cA_{\Delta} :=\big\{(\bn_{t_n})_{n\in\{0,\ldots,N-1\}} \in \cA^0_{\Delta} \, |\,  
	\mathbb{E}[\vert L_{t_n} \vert]<\infty \text{ for all } n\in\{0,\ldots,N-1\}\big\}.
	\end{align*}
A dynamic risk measure (in discrete time) $(\psi_{t_n})_{n\in\{0,\ldots,N-1\}}$ is a family of maps $\psi_{t_n}$ with
\begin{align*}
\psi_{t_n} \colon \mathcal{N}^{\Delta}_{t_n}  \to \mathcal{L}^1(\Omega,\cG_{t_n},\P).
\end{align*}
We restrict to risk measures of the form $\psi_{t_n}(\Loss_{t_n}) = \widetilde  \psi(t_n,X^{\bn}_{t_n},\phe_{t_n},\eta_{t_n})$ for all $n \in \{0,\ldots,N-1\}$ and some measurable function $\widetilde  \psi \colon [0,T]\times\R_+\times [0,\infty)^{d+1} \to \R$.
Lemma \ref{lem:risk_measures_discrete} below shows how the dynamic risk measures VaR, TCE and EL can be computed explicitly in the discrete-time case if we consider  a market with a single stock. Here we can benefit from the fact that the wealth $X^{\bn}_{t_{n+1}}$ at time $t_{n+1}$ is - conditionally on $\cG_{t_n}$ - (shifted) log-normally distributed. In the case $d>1$ the wealth $X^{\bn}_{t_{n+1}}$ at time $t_{n+1}$ given $\cG_{t_n}$ is a sum of dependent log-normally distributed random variables and closed-form expressions for these risk measures are not available.
\begin{lemma}\label{lem:risk_measures_discrete}\leavevmode
	\begin{enumerate}
		\item Given a probability level $\alpha \in (0,1)$ the Dynamic Value at Risk at time $t_n$, $n=0,\ldots,N-1$ can be written as $\VaR^{\alpha}_{t_n}(\Loss_{t_n}) =\widetilde \psi(t_n,X_{t_n}^{\bn},\varphi_{t_n},\eta_{t_n})$, where
		\begin{align*}
		\widetilde  \psi(t,x,\bar\varphi,\bar \eta)
		= {\widetilde  f(t,x)}\!- e^{r\Delta} (x-\bar\eta-\bar\varphi) \!-\exp\big\{\Phi^{-1}(\alpha)\sigma\sqrt{\Delta}  +\big(\mu\!-\!\frac{\sigma^2}{2}\big)\Delta\big\}\bar\varphi.
		\end{align*}
		Recall, $\Phi(\cdot)$ and $\Phi^{-1}(\cdot)$ denote the normal distribution and the inverse distribution functions respectively.
		\item Given a probability level $\alpha \in (0,1)$ the Dynamic Tail Conditional Expectation at time $t_n$, $n=0,\ldots,N-1$ can be written as $\TCE^{\alpha}_{t_n}(\Loss_{t_n}) =\widetilde \psi(t_n,X_{t_n}^{\bn},\varphi_{t_n},\eta_{t_n})$, where
		\begin{align*}
		\widetilde  \psi(t,x,\bar\varphi,\bar \eta)= {\widetilde  f(t,x)} - e^{r\Delta} (x-\bar\eta-\bar\varphi)-\frac{1}{\alpha}e^{ \mu\Delta}\Phi\left(\Phi^{-1}(\alpha)-\sigma\sqrt{\Delta}\right)\bar\varphi.
		\end{align*}
		\item The Dynamic Expected Loss at time $t_n$, $n=0,\ldots,N-1$ can be written as $\EL_{t_n}(\Loss_{t_n}) = \widetilde \psi(t_n,X_{t_n}^{\bn},\varphi_{t_n},\eta_{t_n})$, where
		\begin{align*}
		\widetilde  \psi(t,x,\bar\varphi,\bar \eta) &= \left({\widetilde  f(t,x)} - e^{r\Delta} (x-\bar\eta-\bar\varphi)\right)\Phi(d_1)-e^{ \mu\Delta}\Phi(d_2)\bar\varphi,\\
		d_{1/2} &= \frac{1}{\sigma\sqrt{\Delta}}\Big[ \ln\Big(\frac{\widetilde  f(t,x)- e^{r\Delta} (x-\bar\eta-\bar\varphi)}{\bar\varphi}\Big)-\big(\mu\mp\sigma^2/2\big)\Delta \Big].
		\end{align*}
	\end{enumerate}
\end{lemma}
\begin{proof}
	The proof is presented in Appendix \ref{appB}.
\end{proof}
A risk constraint is imposed on the strategy by requiring that at the beginning of each period (i.e., at $t_n$, $n \in \{0,\ldots,N-1\}$) the investor must decide how much of the wealth is invested in the stocks ($\phe_{t_n}$) and how much is consumed ($\eta_{t_n}$) such that $\bn_{t_n}=(\phe_{t_n}, \eta_{t_n}) \in \cK_D ^R(t_n,X^{\bn}_{t_n})$ with
\begin{align*}
\cK_D ^R(t_n,x)=\left\{\bar \bn=(\bar\phe,\bar \eta) \in \cK_D (x) \,\,\big| \,\,\widetilde  \psi(t_n,x,\bar \phe,\bar \eta) \leq \widetilde  \epsilon(t_n,x)\right\}, \qquad x>0.
\end{align*}
Here, $\widetilde  \epsilon \colon [0,T]\times\R_+ \to [0,\infty)$ is a measurable function which represents the bound on the risk constraint and may depend on time and wealth. Then, the set of admissible strategies reads as
\begin{align*}
\cA_{D}^R :=\left\{(\bn_{t_n})_{n\in\{0,\ldots,N-1\}} \in \cA_{\Delta} \, |\,  \bn_{t_n}\in \cK^R_{\Delta}(t_n,X^{\bn}_{t_n}),\, n\in \{0,\ldots,N-1\}\right\}.
\end{align*}
We assume that the benchmark $\widetilde f(t_n,X^{\bn}_{t_n})$ and bound $\widetilde  \epsilon(t_n,X^{\bn}_{t_n})$ are specified in a way that $\cA_{D}^R \neq \varnothing$ is satisfied,  i.e., there exist admissible strategies.
\subsection{Optimization under risk constraints}
Given an initial wealth $X_0 =x_0>0$ the investor's investment-consumption problem is to decide how much wealth is invested in the stocks and how much is consumed so that the expected value of the utility from consumption and terminal wealth,
\begin{align*}
\E_{t_0,x_0}\Big[\sum_{n=0}^{N-1} U_1(\eta_{t_n}) + U_2\left(X^{\bn}_{t_N}\right)\Big],
\end{align*}
is maximized over all strategies $\bn=(\phe,\eta) \in \cA^R_{\Delta}$. This problem can be solved by using
the theory of Markov Decision Processes, see e.g. \textsc{Bäuerle} and \textsc{Rieder} \cite[Chapter
2]{bau}. The value function is defined by
\begin{align*}
V(t_n,x) = \sup_{(\phe,\eta) \in \cA_{D}^R} \E_{t_n,x}\Big[\sum_{k=n}^{N-1} U_1(\eta_{t_k}) + U_2\left(X^{\bn}_{t_N}\right)\Big]
\end{align*}
and the sequence $(V(t_n,x))_{n=0,\ldots,N-1}$ can be computed by the optimality equation
\begin{align}
\begin{split}\label{DPP}
V(t_N,x) &= U_2(x),\\
V(t_n,x) &= \hspace*{-1ex}\sup\limits_{(\bar \phe,\bar \eta) \in \cK_D ^R(t_n,x)}\hspace*{-1ex}\left\{U_1(\bar \eta)+ \E_{t_n,x}\left[V\left(t_{n+1},e^{r\Delta}\left(x- \bar \eta +\bar\phe'\cdot \Rb_{t_{n+1}}\right)\right)\right]\right\} \\
&\hspace*{50mm} n=N-1,\ldots,0.
\end{split}
\end{align}
The optimal strategy $\bn^{\ast}=(\phe^{\ast},\eta^{\ast})$ is generated by the sequence of maximizers of $V(t_1,x),\ldots,V(t_N,x)$.\\
For the case of power utility, i.e.,
\begin{align*}
U_1(x)=U_2(x)=\frac{x^{1-\gamma}}{1-\gamma}, \quad \gamma \in (0,1),
\end{align*}
we can specify the recursion above. Let $\zeta$ and $\bbeta$ denote the consumption and investment proportion, respectively, i.e.,
\begin{align*}
\zeta_{t_n} := \frac{\eta_{t_n}}{X^{\bn}_{t_n}} \quad \text{and}\quad\beta^i_{t_n} := \frac{\varphi^i_{t_n}}{X^{\bn}_{t_n}-\eta_{t_n}}, \quad \text{ for } n=0,\ldots,N-1, \, i=1,\ldots,d.
\end{align*}
Note that $\zeta_{t_n} \in [0,1]$ and $\bbeta_{t_n}\in \mathcal{P}$ where  $\mathcal{P}:=\{\mathbf{p}\in [0,1]^d~|~ p^1+\ldots+p^d\le 1\}$ denotes the simplex in $\R^d$. In order to use proportions instead of amounts of consumption and investment we redefine $\widetilde \psi$ by
\begin{align*}
\widetilde\psi_{rel}(t_n,X^{\bn}_{t_n},\bbeta_{t_n},\zeta_{t_n}) = \widetilde\psi(t_n,X^{\bn}_{t_n},(1-\zeta_{t_n})X^{\bn}_{t_n}\bbeta_{t_n},\zeta_{t_n}X^{\bn}_{t_n}).
\end{align*}  
Then the risk measures can be written as $\psi_{t_n}(\Loss_{t_n}) = \widetilde  \psi_{rel}(t_n,X^{\bn}_{t_n},\bbeta_{t_n},\zeta_{t_n})$ for all $n \in \{0,\ldots,N-1\}$ and we obtain by the optimality equation \eqref{DPP} the following backward recursion.
\begin{theorem}\label{DP Power}
	The value function for $n=0,\ldots,N$ can be computed by
	\begin{align*}
	V(t_n,x) = \frac{x^{1-\gamma}}{1-\gamma}\cdot d_{t_n} , \quad x> 0,
	\end{align*}
	where $(d_{t_n})$ satisfies the backward recursion
	\begin{align}
	\begin{split}\label{dp_power}
	d_{t_N}&=1,\\
	d_{t_n} &= \sup_{0\leq\bar\zeta\leq 1}\Big\{\bar{\zeta}^{1-\gamma}+(1-\bar \zeta)^{1-\gamma} e^{r\Delta (1-\gamma)} \times \\
	&\hspace*{5mm}\Big(\sup_{\bar \bbeta \in \mathcal{B}(t_n,x,\bar \zeta)}\E\big[\big(1+\bar \bbeta' \cdot \Rb_{t_{n+1}}\big)^{1-\gamma}\big]\Big) d_{t_{n+1}} \Big\},
	~  {n=N-1,\ldots,0,}
	\end{split}
	\end{align}
	where
	\begin{align*}
	\mathcal{B}(t_n,x,\bar \zeta) = \{\bar \bbeta \in  \mathcal{P} \,\,\big|\,\, \widetilde  \psi_{rel}(t_n,x,\bar \bbeta,\bar \zeta) \leq \widetilde  \epsilon(t_n,x)\}.
	\end{align*}
\end{theorem}
\begin{proof}
	The proof is given by mathematical  induction. The basis step is to show that the statement
	\begin{align*}
	V(t_N,x) = \frac{x^{1-\gamma}}{1-\gamma}\cdot d_{t_N},\qquad d_{t_N}=1
	\end{align*}
	holds. This follows straightforwardly from the optimality equation \eqref{DPP}. In the inductive step we assume that the statement
	\begin{align*}
	V(t_n,x) = \frac{x^{1-\gamma}}{1-\gamma}\cdot d_{t_n}
	\end{align*}
	holds for some $n \in \{1,\ldots,N\}$ and show that it also holds for $n-1$. This can be done as follows. From the optimality equation \eqref{DPP} we obtain
	\begin{align*}
	V(t_{n-1},x) = \sup\limits_{(\bar \phe,\bar \eta) \in \cK_D ^R(t_{n-1},x)}\Big\{&\E_{t_{n-1},x}\left[V\left(t_{n},e^{r\Delta}\left(x- \bar \eta +\bar\phe'\cdot \Rb_{t_{n}}\right)\right)\right]\\
	&+\frac{\bar \eta^{1-\gamma}}{1-\gamma}\Big\}.
	\end{align*}
	Using the inductive hypothesis, the right-hand side can be rewritten as
	\begin{align*}
	\sup\limits_{(\bar \phe,\bar \eta) \in \cK_D ^R(t_{n-1},x)}\Big\{ \E_{t_{n-1},x}\Big[\frac{1}{1-\gamma}\left(e^{r\Delta}\left(x- \bar \eta +\bar\phe'\cdot \Rb_{t_{n}}\right)\right)^{1-\gamma}d_{t_n}\Big]
	+\frac{\bar \eta^{1-\gamma}}{1-\gamma}\Big\}.
	\end{align*}
	Substituting $\bar \eta=\bar \zeta x$ and $\bar \phe=(1-\bar \zeta)x\bar \bbeta$ we find 
	\begin{align*}
	\sup\limits_{\substack{0\leq \bar \zeta \leq 1\\ \bar \bbeta \in \mathcal{B}(t_{n-1},x,\bar \zeta)}}\Big\{& \E_{t_{n-1},x}\left[\frac{1}{1-\gamma}\left(e^{r\Delta}\left(x- \bar \zeta  x + x(1-\bar{\zeta})\cdot \bar\bbeta' \Rb_{t_{n}}\right)\right)^{1-\gamma}d_{t_n}\right]\\
	&+\frac{(\bar \zeta x)^{1-\gamma}}{1-\gamma}\Big\}\\
	=\sup\limits_{\substack{0\leq \bar \zeta \leq 1\\ \bar \bbeta \in \mathcal{B}(t_{n-1},x,\bar \zeta)}} \Big\{& \frac{x^{1-\gamma}}{1-\gamma}e^{r\Delta(1-\gamma)}(1- \bar \zeta)^{1-\gamma}\E\left[(1 +\bar\bbeta' \Rb_{t_{n}})^{1-\gamma}\right]d_{t_n}\\
	&+\frac{ x^{1-\gamma}}{1-\gamma}\bar \zeta^{1-\gamma}\Big\}\\
	=\,{\frac{x^{1-\gamma}}{1-\gamma}} \sup\limits_{0\leq \bar \zeta \leq 1 }\Big\{&\bar \zeta^{1-\gamma}+ e^{r\Delta(1-\gamma)}(1- \bar \zeta)^{1-\gamma}\\
	&\sup\limits_{\bar \bbeta \in \mathcal{B}(t_{n-1},x,\bar\zeta)}\E\left[(1 +\bar\bbeta' \Rb_{t_{n}})^{1-\gamma}\right]\cdot d_{t_{n}}\Big\}.
	\end{align*}
	This shows  that indeed it holds $V(t_{n-1},x) = \frac{x^{1-\gamma}}{1-\gamma}\cdot d_{t_{n-1}}$.
	In the second equation we have used that the relative discounted return $\Rb_{t_n}$ is independent of $\cG_{t_{n-1}}$ and in the third equation we used that $e^{r\Delta(1-\gamma)}(1- \bar \zeta)^{1-\gamma}$ is non-negative. Hence, the supremum over $(\bar\bbeta, \bar\zeta)$ can be obtained by the iterated supremum as given.
	The proof is complete by mathematical induction.
\end{proof}
Note that in contrast to the continuous-time case, there is no closed-form solution to the unconstrained problem in discrete time. The Merton portfolio-proportion strategy $(\bbeta^M_{t_n}, \zeta^M_{t_n})_{n=0,\ldots,N-1}$ also has to be computed by backward recursion as in the above theorem where the set $\mathcal{B}$ is replaced by the simplex $\mathcal{P}$, cf. \cite{bau}. Then, we obtain
\begin{align*}
\bbeta_{t_n}^M= \arg\max_{\bar \bbeta \in \mathcal{P}}\E\left[(1 +\bar\bbeta' \Rb_{t_{n+1}})^{1-\gamma}\right] \, \text{ and }\, \zeta^M_{t_n}=\Big[1+\big(e^{r\Delta(1-\gamma)}v_{t_n}d_{t_n}\big)^{\frac{1}{\gamma}}\Big]^{-1}
\end{align*}	
for $n=0,\ldots,N-1$ with 
\begin{align*}
v_{t_n} :=\sup_{\bar \bbeta \in \mathcal{P}}\E\big[\big(1+\bar \bbeta' \cdot \Rb_{t_{n+1}}\big)^{1-\gamma}\big], \quad n=0,\ldots,N-1.
\end{align*}	
The value function and wealth obtained by an investor using the Merton portfolio-proportion strategy are denoted by $V^M$ and $X^{M}$, respectively.	
\section{Numerical examples for power utility}\label{sec4}
In this section we solve the continuous-time optimization problem
\begin{align*}
\sup_{(\pe,c)\in \cA^R}\E_{0,x_0}\Big[\int\nolimits
_0^T U_1(C_t)\d t + U_2\left(X^{(\pe,c)}_T\right)\Big]
\end{align*}
and the discrete-time optimization problem
\begin{align*}
\sup_{(\phe,\eta) \in \cA_{D}^R}\E_{0,x_0}\Big[\sum_{n=0}^{N-1} U_1(\eta_{t_n}) + U_2\left(X^{(\phe,\eta)}_{t_N}\right)\Big]
\end{align*}
numerically for the case of power utility and compare the solutions. Our numerical experiments are based on the following model parameters. The financial market consists of a bond with risk-free interest rate $r=0.1$ and a single stock with drift $\mu = 0.18$ and volatility $\sigma =0.35$. If not stated otherwise, the parameter of the power utility function is $\gamma = 0.3$ and the terminal trading time is $T=2$ years. In the continuous-time case the dynamic risk measures are evaluated under the assumption that the portfolio-proportion process is kept constant between $t$ and $t+\Delta$ with $\Delta = \frac{1}{24} \approx 2$ weeks. The probability level in the definition of the Value at Risk and Tail Conditional Expectation is given by $\alpha = 0.01$.
\subsection{Continuous-time optimization problem}
We start with the continuous-time problem by numerically solving the HJB equation \eqref{HJB}. Using a policy improvement (PI) algorithm, we obtain an approximation of the value function and of the optimal portfolio-proportion process. In each iteration of the PI algorithm, we have to solve a linear partial differential equation (PDE) and a constrained optimization problem. Both, the linear PDE and the constrained optimization problem, are solved numerically. The former by using meshless methods (see \cite[Chapter 10]{ji} and \cite[Chapter 16]{duf}), the latter with sequential quadratic programming methods (see \cite[Chapter 18]{noc}).\\
\begin{figure}[h]
	\centering
	\scalebox{0.31}{\subfigure{\includegraphics{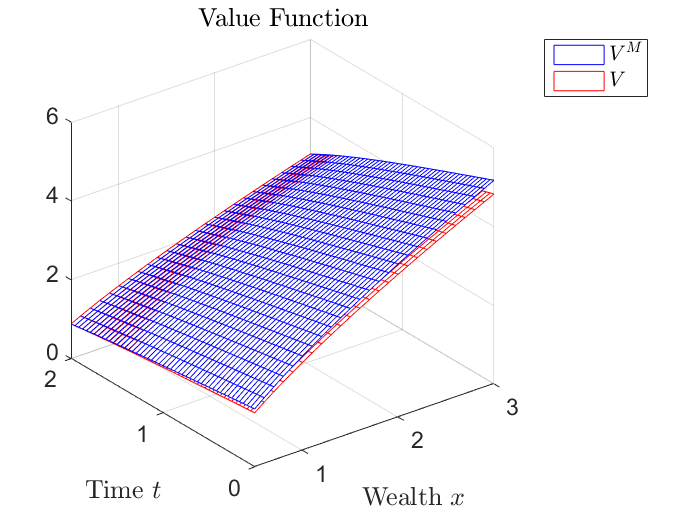}}}
	\scalebox{0.31}{\subfigure{\includegraphics{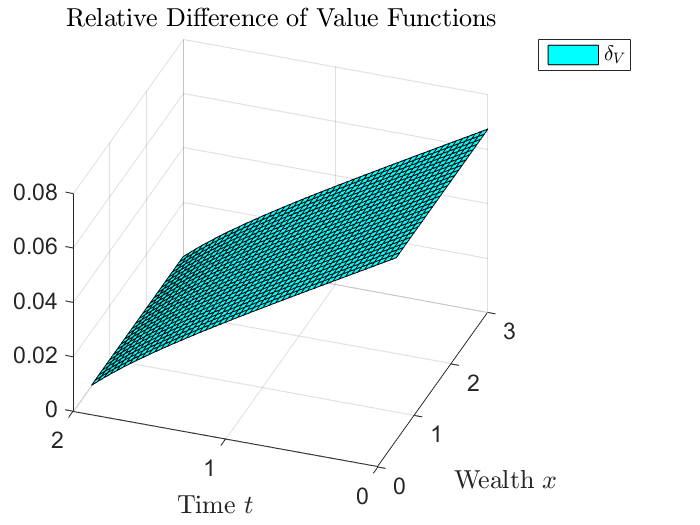}}}
	\scalebox{0.31}{\subfigure{\includegraphics{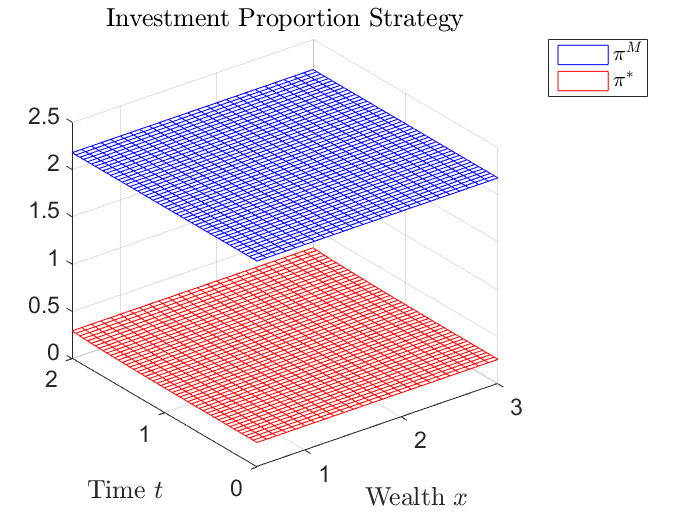}}}
	\scalebox{0.31}{\subfigure{\includegraphics{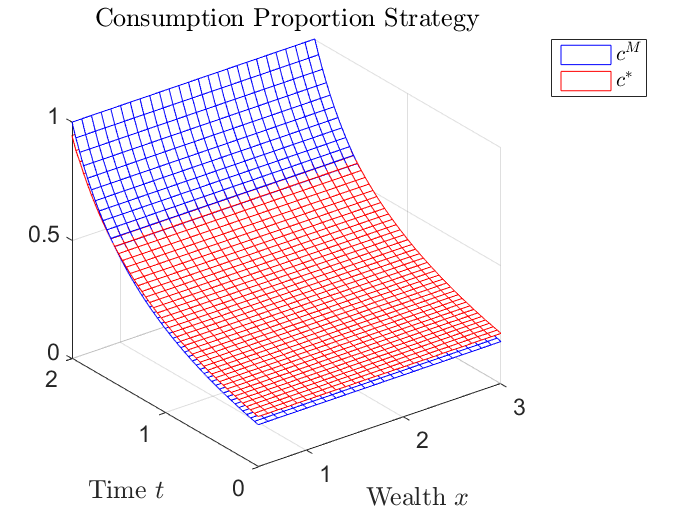}}}
	\captionsetup{justification=raggedright}
	\caption{Effect of the VaR constraint on the value function and optimal portfolio-proportion strategy for the continuous-time problem. The benchmark and bound are given by $Y_t^M$ and $\widetilde  \varepsilon(t,x) = 0.05 x$, respectively.}\label{fig:1}
\end{figure}
Figure \ref{fig:1} shows the effect of the VaR constraint on the value function and optimal portfolio-proportion process. The benchmark for the VaR is the conditional expectation of wealth $X_{t+\Delta}^M$ obtained by an investor following the Merton portfolio-proportion process $(\pi_t^M,c_t^M)$ (cf. \eqref{Merton_strategy}) in $[t, t+\Delta]$ given the wealth $X_t^{(\pi,c)}$ at time $t$, i.e.,
\begin{align*}
Y_t^M := \E[X^{M}_{t+\Delta}|X_t^{(\pi,c)}] = X_t^{(\pi,c)}\exp\left\{ \left(r + \pi_t^M(\mu -r)-c_t^M\right) \Delta \right\}.
\end{align*}
The bound for the VaR is given by $\widetilde  \varepsilon(t,x) = 0.05 x$. This is a bound which is relative to the wealth, i.e., any loss in the interval $[t,t+\Delta]$ can be hedged with 5\% of the portfolio value. By comparison, the risk of the Merton portfolio-proportion process $(\pi_t^M,c_t^M)$ which is held constant in $[t,t+\Delta]$ is $\widetilde  \varepsilon(t,x) \approx 0.31 x$. A first look at Figure \ref{fig:1} indicates that the value function and the relative consumption are not remarkably affected by the VaR constraint whereas the proportion of wealth invested in the risky stock is considerably reduced. The top, right-hand panel of Figure \ref{fig:1} shows the relative difference $\delta_{V}(t,x)$ between the value function of a Merton investor and a VaR-constrained investor defined by  
	\begin{align*}
	\delta_{V}(t,x)=\frac{V^M(t,x)-V(t,x)}{V^M(t,x)}.
	\end{align*}
	In order to facilitate the comparison of the value functions we express the losses of performance due to the risk constraint in monetary units  and 	
	introduce the following efficiency measure. The efficiency of an investor $A$ relative to an investor $B$ is the initial amount of wealth that investor $B$ would need to obtain a value function identical to that of investor $A$ who started at time $t=0$ with unit wealth. Figure \ref{fig:eff_con} illustrates the efficiency of a VaR-constrained investor relative to a Merton investor in continuous time for different relative bounds of the VaR-constraint. As expected, the efficiency increases when the bound of the risk constraint becomes less restrictive. For the bound $\widetilde \varepsilon(t,x)=0.05x$ used in Fig.~\ref{fig:1} the loss of efficiency is about 9.5\%. For the most restrictive case, where wealth below the benchmark is not tolerated at all, i.e., the bound is set to $\widetilde  \varepsilon(t,x) = 0$, the loss of efficiency is  12.6\%.\\
\begin{figure}[h]
	\centering
	\scalebox{0.31}{\includegraphics{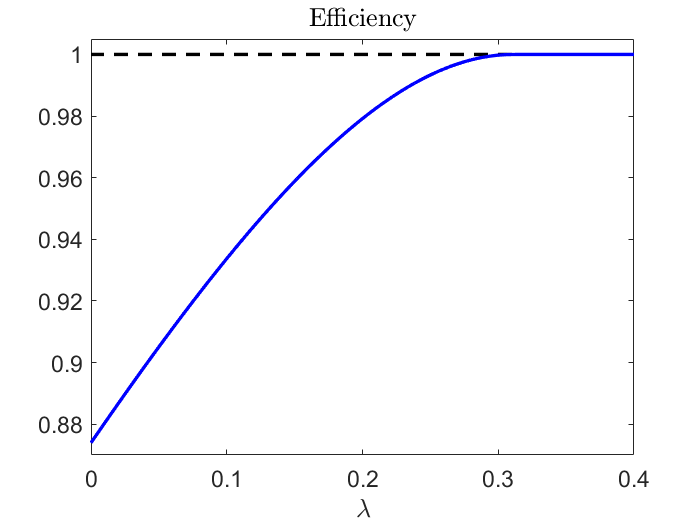}}
	\captionsetup{justification=raggedright}
	\caption{Efficiency of a VaR-constrained investor relative to a Merton investor in continuous time for different relative bounds $\widetilde  \varepsilon(t,x) =  \lambda x$. The benchmark is given by $Y^M_t$.} \label{fig:eff_con}
\end{figure} 
Using a relative bound leads to a constant proportion of wealth invested in the risky stock as in the unconstrained case. If we change the relative bound to an absolute one, e.g. $\widetilde  \varepsilon(t,x)\equiv 0.05$, the optimal proportion of wealth invested in the risky stock is no longer independent of the wealth level, cf. Figure \ref{fig:3}. This results from the fact that for $x<1$, the absolute bound is less restrictive and for $x>1$, it is more restrictive than the relative one. For the optimal relative consumption  no remarkable differences between an absolute and a relative bound are observed.
\begin{figure}[!h]
	\centering
	\scalebox{0.31}{\subfigure{\includegraphics{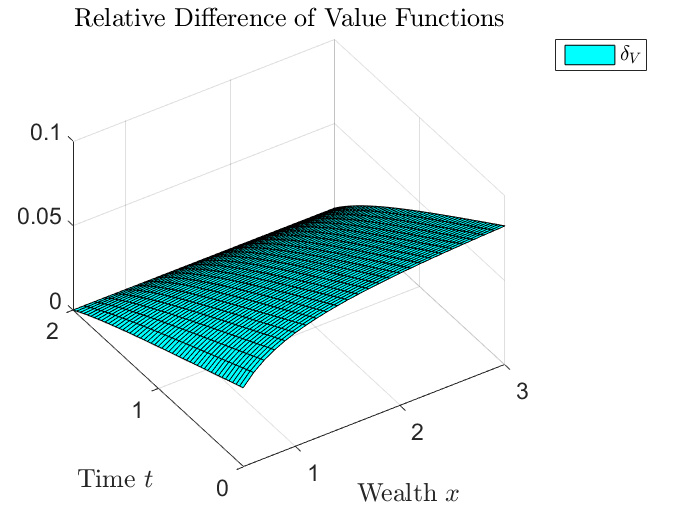}}}
	\scalebox{0.31}{\subfigure{\includegraphics{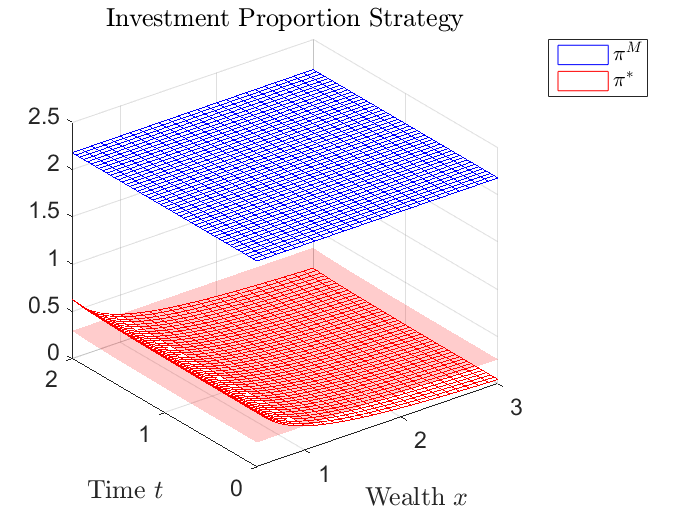}}}
	\captionsetup{justification=raggedright}
	\caption{Effect of the VaR constraint on the value function and optimal investment proportion strategy for the continuous-time problem. The benchmark and bound are given by $Y_t^M$ and $\widetilde  \varepsilon(t,x) = 0.05$, respectively.}\label{fig:3}
\end{figure}
\begin{figure}[h!]
	\centering
	\hspace*{-5mm}
	\scalebox{1.00}{\subfigure{\includegraphics[width=0.51\textwidth]{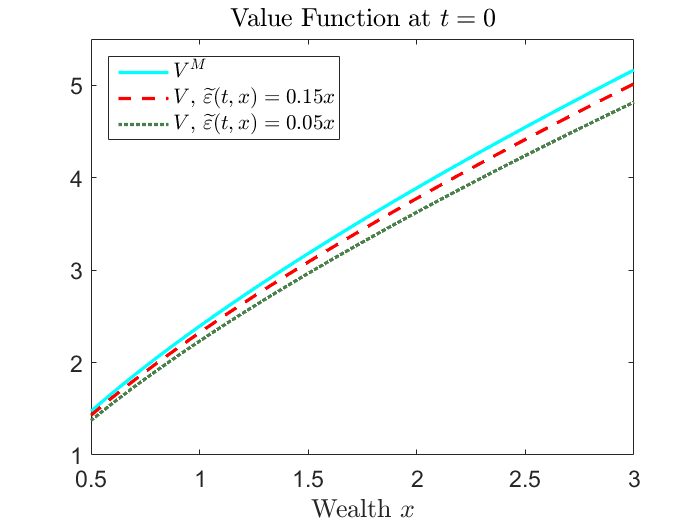}}}
	\hspace*{-5mm}
	\scalebox{1.00}{\subfigure{\includegraphics[width=0.51\textwidth]{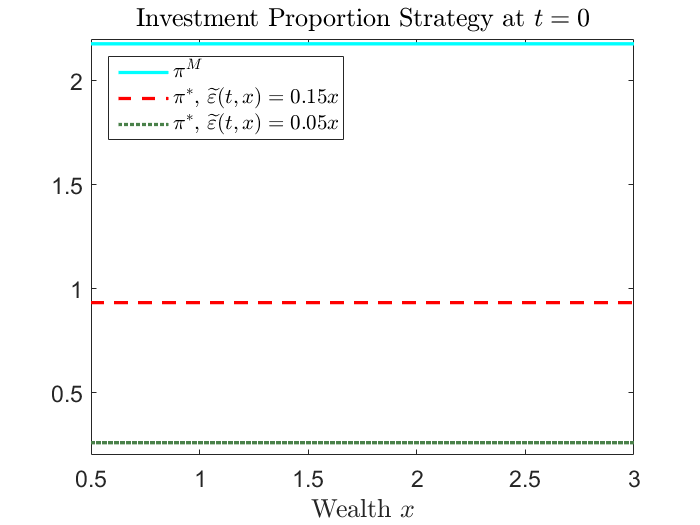}}}
	
	\vspace*{-3mm}
	\hspace*{-5mm}
	\scalebox{1.00}{\subfigure{\includegraphics[width=0.51\textwidth]{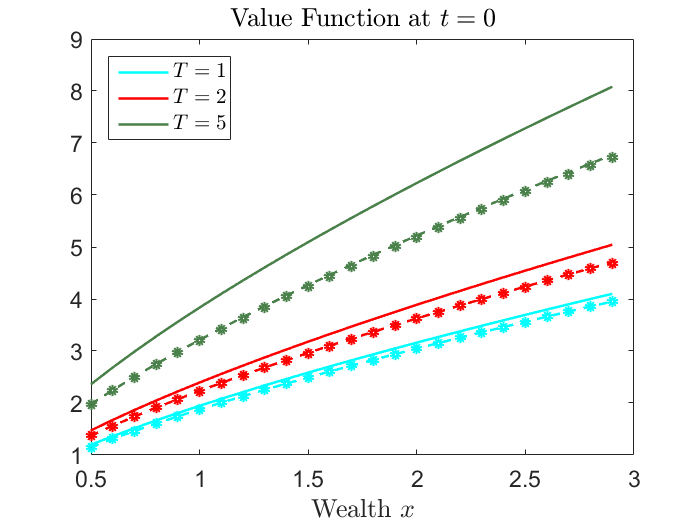}}}
	\hspace*{-5mm}
	\scalebox{1.00}{\subfigure{\includegraphics[width=0.51\textwidth]{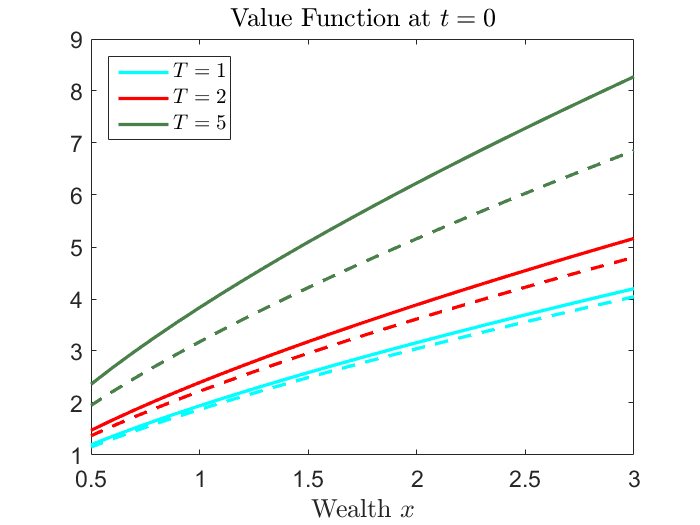}}}
	\captionsetup{justification=raggedright}
	\caption{\small{\textbf{Top panels:} Effect of different relative bounds ($\tilde \varepsilon(t,x)=0.15x,0.05x$) for the VaR constraint on the value function and optimal investment proportion strategy in the continuous-time case. The benchmark and terminal trading time are given by $Y_t^M$ and $T=2$, respectively.\\
			\textbf{Bottom panels:} Effect of different terminal trading times ($T=1,2,5$) and different risk constraints (VaR, TCE, EL) on the value function in the continuous-time case. The benchmark is given by $Y_t^M$.\\
			\textbf{Bottom left panel:} Value function of the Merton investor (solid line), the VaR-constrained investor (dashed line) and the TCE-constrained investor (asterisk) for $\tilde \varepsilon(t,x)=0.05x$.\\
			\textbf{Bottom right panel:} Value function of the Merton investor (solid line) and the EL-constrained investor (dashed line) for $\tilde \varepsilon(t,x)=0.01x$.}} \label{fig:4}
\end{figure}
From the top panels of Figure \ref{fig:4}, we can observe that a more restrictive bound leads to a smaller investment in the risky stock and a smaller value function. Another interesting observation can be made from the bottom panels of Figure \ref{fig:4}. The value function at time $t=0$ is plotted against the wealth level for different terminal trading times $T$. The left panel shows the value function of a Merton investor (solid line),  a VaR-constrained investor (dashed line) and a TCE-constrained investor (asterisk). In the right panel, the value function of a Merton investor (solid line) and of a EL-constrained investor (dashed line) is plotted. We observe that the value function is not noticeably affected by the choice of different risk constraints and that there are almost no differences between the VaR constraint and the TCE constraint. This confirms similar results by \cite{cuo}. Even though a static VaR constraint has been found to induce an increased probability of extreme losses and an increased allocation to risky assets in some states (see \cite{bs}), these shortcomings vanish if a VaR constraint is imposed dynamically (cf. \cite{cuo}).

For longer trading horizons, the effects of a risk constraint on the value function become more noticeable. We obtain comparable results when we change the benchmark to the wealth an investor will obtain at time $t+\Delta$, starting with $X_t^{(\pi,c)}$ at time $t$, while only investing in the bond (and not in the stock) and consuming the wealth with the rate $c_t^M$ in $[t,t+\Delta]$, i.e.,
$Y_t = X^{(\pi,c)}_{t} \cdot e^{(r-c_t^M)\Delta}.$
\begin{remark}
	We also performed numerical experiments for the case $\gamma > 1$ instead of $\gamma=0.3$ which led to smaller differences between the value function of a Merton investor and a risk-constrained investor. This results from the fact that for a larger $\gamma$ an investor is more risk-averse even without an imposed risk constraint.
\end{remark}
\subsection{Discrete-time optimization problem}
We now consider the discrete-time problem and apply the recursion in Theorem \ref{DP Power}. The expectation in \eqref{dp_power} can be written as
\begin{align*}
\E&\left[\left(1+\bar \beta \cdot R_{t_{n+1}}\right)^{1-\gamma}\right]\\
&=\int_0^{\infty}\Big[ \big(1+\bar \beta (u-1)\big)^{1-\gamma}\cdot \frac{1}{\sqrt{2\pi\Delta}\sigma u}\exp\Big\{-\frac{\big(\ln u-(\mu-r-\frac{\sigma^2}{2})\Delta\big)^2}{2\sigma^2\Delta}\Big\}\Big]\d u,
\end{align*}	
where we used that $R_{t_{n+1}}=\widetilde R_{t_{n+1}}/e^{r\Delta}-1$ and $\widetilde R_{t_{n+1}}$ is log-normally distributed. The above integral is evaluated numerically using quadrature rules.
\begin{figure}[!h]
	\centering
	\scalebox{0.31}{\subfigure{\includegraphics{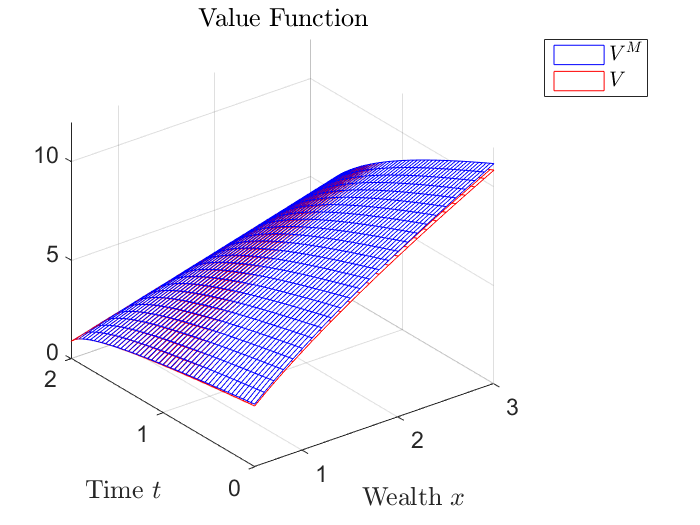}}}
	\scalebox{0.31}{\subfigure{\includegraphics{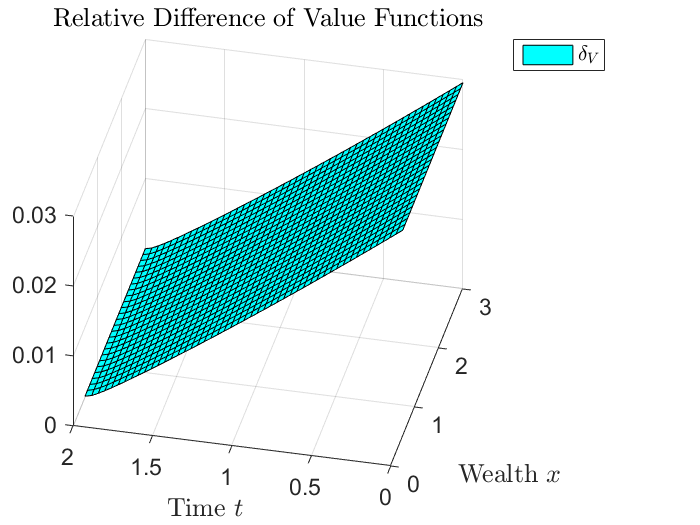}}}	
	\scalebox{0.31}{\subfigure{\includegraphics{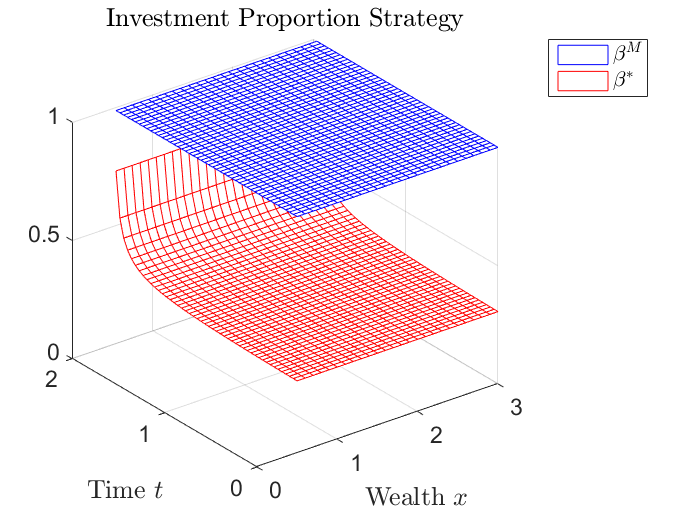}}}
	\scalebox{0.31}{\subfigure{\includegraphics{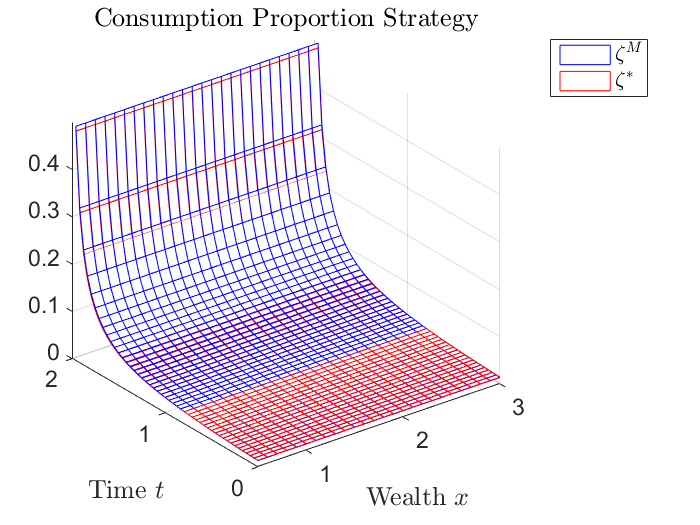}}}
	\captionsetup{justification=raggedright}
	\caption{Effect of the VaR constraint on the value function and optimal portfolio-proportion strategy in the discrete-time case. The benchmark and bound are given by $Y^M_{t_n}$ and $\widetilde \epsilon(t_n,x) = 0.05x$, respectively. }\label{fig:7}
\end{figure}
Figure \ref{fig:7} shows the effect of the VaR constraint on the value function and optimal portfolio-proportion strategy. The benchmark for the VaR is the conditional expectation of wealth $X_{t_{n+1}}^M$ obtained by an investor following the (discrete-time) Merton strategy $\varphi_{t_n}^M,\eta_{t_n}^M$ given the wealth $X_{t_n}^{(\varphi,\eta)}$ at time $t_n$, i.e.,
\begin{align*}
Y^M_{t_n} = \E[X^M_{t_{n+1}}|X_{t_n}^{(\varphi,\eta)}] = e^{r\Delta}\left(X_{t_n}^{(\varphi,\eta)}- \eta^M_{t_n} - \varphi^M_{t_n}\right) + e^{\mu \Delta} \varphi^M_{t_n}.
\end{align*}
The bound for the VaR constraint was given by $\tilde \epsilon(t_n,x)=0.05x$. By comparison, the risk of the Merton strategy $(\varphi_{t_n}^M,\eta_{t_n}^M)$ is $\widetilde  \epsilon(t,x) \approx 0.16 x$. It can be observed that the value function is not remarkably affected by the VaR constraint. Moreover, the effects on the value function are even less notable than in the continuous-time case. This results from the fact that even without an imposed risk constraint short-selling the stock or bond is not allowed, thus the proportion invested in the stock is always in $[0,1]$. Note that in the above example for the continuous-time problem, the Merton investment proportion is $\pi^M\approx 2.18$, i.e.,~it exceeds one. Furthermore, it can be observed from Figure \ref{fig:7} that the fraction of wealth invested in the risky stock is considerably reduced when the VaR constraint is imposed, whereas the differences in the consumption rate between a VaR-constraint investor and a Merton investor are hard to distinguish visually. If we use an absolute bound $\widetilde  \epsilon(t,x)\equiv 0.05$ for the VaR constraint instead of a relative bound, the optimal investment strategy is no longer a constant proportion of wealth and the results are similar to Figure \ref{fig:3} for the continuous-time case. Numerical results for varying the terminal trading time or the risk measure are not shown here, but they are comparable to the continuous-time case, cf.~Figure \ref{fig:4}.\\
As in the continuous-time case we now express the losses of performance of a risk-constrained investor relative to the performance of a Merton investor in monetary units  using the efficiency measure.
	The left panel of Figure \ref{fig:13} illustrates the efficiency of a VaR-constrained investor relative to a Merton investor in discrete time for different relative bounds of the VaR-constraint. As expected and already observed in Fig.~\ref{fig:eff_con}, the efficiency increases when the bound of the risk constraint becomes less restrictive. The loss of efficiency is at most 7.2\% which is attained for the most restrictive bound, i.e., $\widetilde \epsilon(t,x)=0$. For the bound $\widetilde \epsilon(t,x)=0.05x$ used in Fig.~\ref{fig:7} the loss of efficiency is about 4.2\%.
\\
\begin{figure}[!t]
	\centering
	\subfigure{\includegraphics[width=0.49\textwidth]{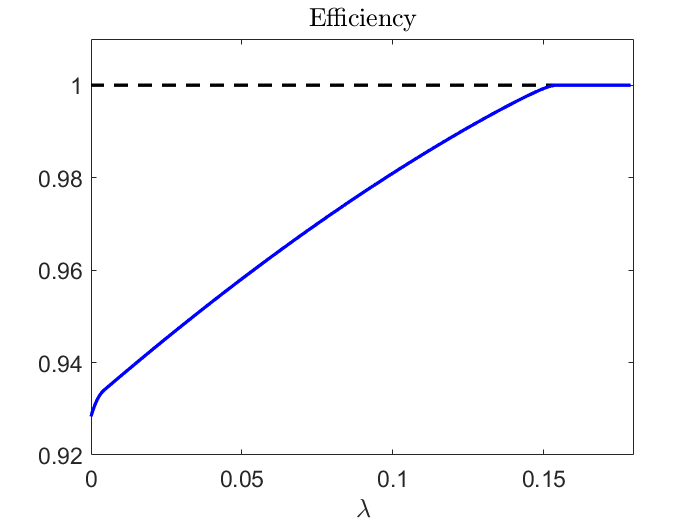}}
	\subfigure{\includegraphics[width=0.49\textwidth]{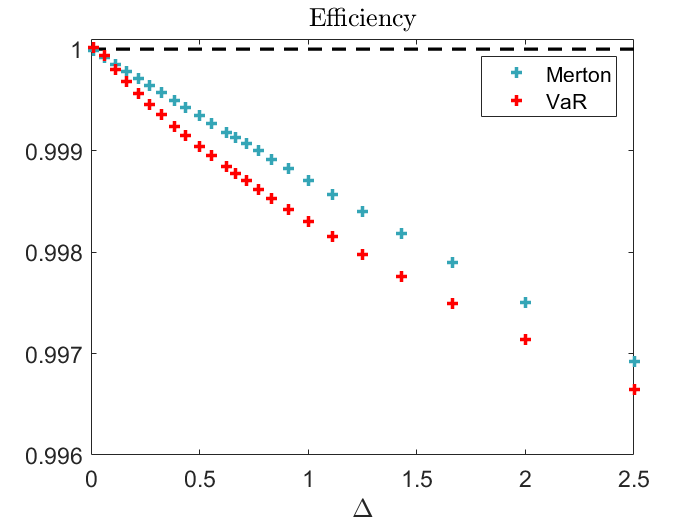}}
	\captionsetup{justification=raggedright}
	\caption{Efficiency for different bounds $\widetilde  \epsilon(t,x) =  \lambda x$ (left panel) and for different risk measurement horizons $\Delta$ (right panel). The benchmark is given by $Y^M_t$. \\
		\textbf{Left panel:} Efficiency of a VaR-constrained investor relative to a Merton investor in discrete time for $\gamma=0.3$ and $\Delta=1/24$.\\
		\textbf{Right panel:} Efficiency of a discrete-time relative to a  continuous-time Merton investor and efficiency of a discrete-time relative to a continuous-time  VaR-constrained investor for $\gamma=0.9$ and $\widetilde  \epsilon(t,x) =0.05x$.}\label{fig:13}
\end{figure}
We finish this section with a numerical comparison of the continuous-time and discrete-time case. The availability of solutions to the risk-constrained portfolio problem both for discrete and continuous time allows us to quantify the losses of portfolio performance resulting from time discretization, i.e.,~from the restriction to finite trading frequencies. First observe that the loss of efficiency resulting from imposing a risk-constraint is higher in continuous time than in discrete time, cf. Figure \ref{fig:eff_con} and left panel of Figure \ref{fig:13}. This arises from the fact that short-selling is allowed in continuous time leading to $\pi^M\approx 2.18 >1$ whereas it is not allowed in discrete time leading to $\beta^M=1$. For a fair comparison we do not allow for short-selling in both cases. Since the parameters of the financial market shall remain the same, we change the investors' preferences represented by the utility functions by setting $\gamma=0.9$. This leads to $\pi^M,\, \beta^M \in [0,1]$, i.e., there is no short-selling in both cases. Furthermore, the expected utility from consumption in the time interval $[t_n,t_{n+1})$, $n=0,\ldots,N-1$, is given by $\E[\int_{t_n}^{t_{n+1}}U_1(C_s) \d s]$ in the continuous-time case. However, in the discrete-time case we have $\E[U_1(\eta_{t_n})]$, which approximately corresponds to $\E[U_1(\int_{t_n}^{t_{n+1}}C_s \d s)]$.  The related maximization problem 
	\begin{align} \label{op2} \nonumber
	\mathbb{E}_{t_0,x_0}\Big[ U_1\Big(\int_0^T C_t \d \,t\Big) + U_2(X^u_T)\Big]
	\end{align} 
	is not only different from the economic interpretation, but also from the mathematical point of view, cf. \textsc{Grandits et al.} \cite{gra}. Since we want to quantify the losses of portfolio performance resulting solely from restricting to finite trading frequencies, we set $U_1(x)\equiv 0$. Again, we use the efficiency to compare the portfolio performance of a discrete-time investor relative to a continuous-time investor with and without imposing a risk constraint. For this experiment, we compute the VaR risk measure for the continuous-time investor under the assumption that the number of shares remains constant between $t$ and $t+\Delta$. In \textsc{Rogers} \cite{rog} the continuous-time investment problem is compared to the discrete-time investment problem when no risk constraint is imposed. The author shows that there is only a small difference between the continuous and discrete problem when the additional restriction that the continuous-time investor is not allowed to sell short the stock or the bond is imposed. \textsc{Bäuerle} et al. \cite{buv} show that this remains true for models where the drift is modeled by a random variable that is not directly observable and has to be estimated from observed stock prices. However, the authors found that when the short-selling restriction for the continuous-time investor is omitted, the discrete-time investor will generally not do as well as the continuous-time investor and a discretization gap remains. The right panel of Figure \ref{fig:13} shows the efficiency of a discrete-time  relative to a continuous-time Merton investor (cyan) and the efficiency of a discrete-time relative to a continuous-time VaR-constrained investor (red). We observe that the values of the efficiency are very close to unity; even for a quite large $\Delta$ in the range of $2.5$ years, the loss of efficiency is at most $0.35\%$ in both cases. This  is in line with the results by \textsc{Rogers} \cite{rog} for the Merton investor. In addition to the results in \cite{rog} our numerical results for the risk-constrained case indicate that the losses due to time discretization are of comparable (small) magnitude.

\bibliographystyle{plain}

\begin{appendix}
	\section{Proof of Lemma \ref{lem:risk_measures}}\label{appA}
	Under the assumption that the portfolio-proportion strategy $(\pe,c)$ is kept constant and equal to $(\bar \pe,\bar c)$ between time $t$ and $t+\Delta$ the wealth at time $t+\Delta$ is given by
	\begin{align*}
	X_{t+\Delta} = \exp\Big\{&\ln\big( X_t\big)+\Big(\bar\pe'(\me-\1 r)+r-\bar c-\frac{\lVert\bar\pe'\se\rVert^2}{2}\Big)\Delta  \\
	&+ \bar\pe'\se (\We_{t+\Delta }-\We_t)\Big\}.
	\end{align*}
	From the equation above we obtain that $X_{t+\Delta}$ is - conditionally on $\cF_t$ - distributed as $e^Z$, where $Z$ is normally distributed with mean $m$ and variance $s^2$. Here,
	\begin{align*}
	m:= \ln\left(X_t\right)+\Big(\bar\pe'(\me-\1 r)+r-\bar c-\frac{1}{2}\lVert\bar\pe'\se\rVert^2\Big)\Delta  \quad \text{and} \quad s^2:= \lVert\bar\pe'\se\rVert^2\Delta.
	\end{align*}
	Recall, the Expected Loss at time $t$ is defined by
	\begin{align*}
	\EL_t\left(\Loss_{t}\right) :=\E\Big[\left(Y_t -X_{t+\Delta}\right)^+ \Big|\,\, \cF_t\Big] =\E\Big[\left(Y_t-e^Z\right)^+\Big|\,\, \cF_t \Big].
	\end{align*}
	This expectation can be calculated as follows. Let
	\begin{align*}
	f_Z(z) = \frac{1}{\sqrt{2\pi}s}\exp{\Big(-\frac{(z-m)^2}{2s^2}\Big)}
	\end{align*}
	denote the probability density function of $Z$ then
	\begin{align*}
	\EL_t\left(\Loss_{t}\right) &= \int\nolimits_{-\infty}^\infty {\left(Y_t-e^z\right)^+ f_Z(z)\d z}=\int\nolimits_{-\infty}^{\ln(Y_t)} {\left(Y_t -e^z\right) f_Z(z)\d z}=Y_t I_1 + I_2,
	\end{align*}
	where $I_1:=\int_{-\infty}^{\ln(Y_t)}  f_Z(z)\d z ~\text{ and }~
	I_2:= -\int_{-\infty}^{\ln(Y_t)} e^z f_Z(z)\d z.$\\
	For the integral $I_1$ an appropriate change of variables yields
	\begin{align*}
	I_1= \int\nolimits_{-\infty}^{d_1} { \frac{1}{\sqrt{2\pi}}e^{-\frac{y^2}{2}}\d y} =  \Phi(d_1),
	\end{align*}
	where
	\begin{align*}
	d_1:&=\frac{\ln(Y_t)-m}{s}=\frac{1}{\lVert\bar\pe'\se\rVert\sqrt{\Delta }}\Big[ \ln\Big(\frac{Y_t}{X_t}\Big)-\big(\bar\pe'(\me-\1 r)+r-\bar c-\frac{\lVert\bar\pe'\se\rVert^2}{2}\big)\Delta \Big]
	\end{align*}
	and $\Phi(\cdot)$ denotes the cumulative distribution function of the standard normal distribution.
	The integral $I_2$ can be written as
	\begin{align*}
	I_2 &= -\int\nolimits_{-\infty}^{\ln(Y_t)} {\frac{1}{\sqrt{2\pi}s}\exp{\Big(z-\frac{z^2-2zm+m^2}{2s^2}\Big)}\d z}\\
	&=-e^{\frac{s^2}{2}+m}\int\nolimits_{-\infty}^{\ln(Y_t)} {\frac{1}{\sqrt{2\pi}s}\exp{\Big(-\frac{\left(z-(m+s^2)\right)^2}{2s^2}\Big)}\d z}.		
	\end{align*}
	Using the change of variables technique yields
	\begin{align*}
	I_2 &=-e^{\frac{s^2}{2}+m}\int\nolimits_{-\infty}^{d_2} {\frac{1}{\sqrt{2\pi}}e^{-\frac{y^2}{2}}\d y}=-e^{\frac{s^2}{2}+m}\Phi(d_2),
	\end{align*}
	where
	\begin{align*}
	d_2:&=\frac{\ln(Y_t)-(m+s^2)}{s}\\
	&=\frac{1}{\lVert\bar\pe'\se\rVert\sqrt{\Delta}}\bigg[ \ln\Big(\frac{\widetilde f(t,X_t)}{X_t}\Big)-\Big(\bar\pe'(\me-\1 r)+r-\bar c+\frac{1}{2}\lVert\bar\pe'\se\rVert^2\Big)\Delta \bigg].		
	\end{align*}
	Note, $\exp\{s^2/2+m\}= X_t\exp{\left(\left(\bar\pe'(\me-\1 r)+r-\bar c\right)\Delta \right)}$. Finally we obtain
	\begin{align*}
	\EL_t\left(\Loss_{t}\right)=Y_t I_1 + I_2 = \widetilde f(t,X_t) \Phi(d_1) -X_t\exp{\left(\left(\bar\pe'(\me-\1 r)+r-\bar c\right)\Delta \right)}\Phi(d_2).
	\end{align*}
	\section{Proof of Lemma \ref{lem:risk_measures_discrete}}\label{appB}
	Given an investment-consumption strategy $(\varphi,\eta)$ the wealth process $X$ evolves as follows
	\begin{align*}
	X_{t_{n+1}} &= e^{r\Delta}(X_{t_n}-\eta_{t_n}-\varphi_{t_n})+\varphi_{t_n}\cdot \exp\big\{\big(\mu-\frac{\sigma^2}{2}\big)\Delta + \sigma \big(W_{t_{n+1}}-W_{t_n}\big) \big\}.
	\end{align*}
	Note that given a probability level $\alpha \in (0,1)$ the Value at Risk at time $t_n$ is defined by
	$
	\VaR_{t_n}^\alpha(\Loss_{t_n}) := \inf\left\{l \in \R \,|\, \P(\Loss_{t_n} \geq l|\cG_{t_n}) \leq \alpha \right\}.
	$ 
	We have
	\begin{align*}
	\Loss_{t_n} = Y_{t_n} - X_{t_{n+1}}&=\widetilde f(t_n,X_{t_n}) - e^{r\Delta}(X_{t_n}-\eta_{t_n}-\varphi_{t_n})\\
	&-\varphi_{t_n}\cdot \exp\Big\{\Big(\mu-\frac{\sigma^2}{2}\Big)\Delta + \sigma \left(W_{t_{n+1}}-W_{t_n}\right)\Big\}
	\end{align*}
	and
	\begin{align*}
	\P\left(\Loss_{t_n} \geq l|\cG_{t_n}\right)&=\P\Big(\exp\Big\{\Big(\mu-\frac{\sigma^2}{2}\Big)\Delta + \sigma \left(W_{t_{n+1}}-W_{t_n}\right) \Big\}\\
	&\qquad \leq \frac{\widetilde f(t_n,X_{t_n}) - e^{r\Delta}(X_{t_n}-\eta_{t_n}-\varphi_{t_n}) - l}{\varphi_{t_n}}\Big|\cG_{t_n} \Big)\\
	&=\P\big(\Delta^{-\frac{1}{2}}\big(W_{t_{n+1}}-W_{t_n}\big) \leq z \big|\cG_{t_n} \big)=\Phi(z),
	\end{align*}
	where
	\begin{align*}
	z = \frac{1}{\sigma \sqrt{\Delta}} \Big(\ln\Big\{ \frac{\widetilde f(t_n,X_{t_n}) - e^{r\Delta}(X_{t_n}-\eta_{t_n}-\varphi_{t_n}) - l}{\varphi_{t_n}} \Big\} - \big(\mu-\frac{\sigma^2}{2}\big)\Delta\Big).
	\end{align*}
	We have used that the random variable $\Delta^{-\frac{1}{2}}(W_{t_{n+1}}-W_{t_n})$ is standard normally distributed and independent of $\cG_{t_n}$. Thus, $\P\left(\Loss_{t_n} \geq l|\cG_{t_n}\right) = \Phi(z) \leq \alpha$ is satisfied for $z \leq \Phi^{-1}(\alpha)$ yielding
	\begin{align*}
	l \geq \widetilde f(t_n,X_{t_n}) - e^{r\Delta}(X_{t_n}-\eta_{t_n}-\varphi_{t_n})-\exp\big\{\Phi^{-1}(\alpha)\sigma \sqrt{\Delta}+\big(\mu-\frac{\sigma^2}{2}\big)\Delta \big\}\varphi_{t_n}.
	\end{align*}
	Since the Dynamic Value at Risk is the smallest $l$ satisfying the above inequality we obtain $\VaR^{\alpha}_{t_n}(\Loss_{t_n}) = \widetilde \psi(t_n,X_{t_n},\varphi_{t_n},\eta_{t_n})$, where
	\begin{align*}
	\widetilde  \psi(t,x,\bar\varphi,\bar \eta)= &\widetilde f(t,x)- e^{r\Delta} (x-\bar \eta-\bar\varphi)-\exp\Big\{\Phi^{-1}(\alpha)\sigma\sqrt{\Delta} +\Big(\mu-\frac{\sigma^2}{2}\Big)\Delta\Big\}\bar\varphi.
	\end{align*}
	The Tail Conditional Expectation at time $t_{n}$ is defined by
	\begin{align*}
	\TCE_{t_n}^\alpha(\Loss_{t_n}) &= \E_{t_n}\left[\Loss_{t_n}|\Loss_{t_n} \geq \VaR_{t_n}^\alpha(\Loss_{t_n})\right]\\
	& = \frac{\E\left[\Loss_{t_n}\boldsymbol{I}\left(\Loss_{t_n} \geq \VaR_{t_n}^\alpha(\Loss_{t_n})\right)\big|\cG_{t_n}\right]}{\P\left(\Loss_{t_n} \geq \VaR_{t_n}^\alpha(\Loss_{t_n})\big|\cG_{t_n}\right)}\\
	&=\frac{1}{\alpha} \E\left[\Loss_{t_n}\boldsymbol{I}\left(\Loss_{t_n} \geq \VaR_{t_n}^\alpha(\Loss_{t_n})\right)\big|\cG_{t_n}\right],
	\end{align*}
	where $\boldsymbol{I}(A)$ denotes the indicator function of the set $A$.
	Using
	\begin{align*}
	\Loss_{t_n}=Y_{t_n}- X_{t_{n+1}}=&\widetilde f(t_n,X_{t_n})-e^{r\Delta}(X_{t_n}-\eta_{t_n}-\varphi_{t_n})\\
	&-\varphi_{t_n} \exp\Big\{\Big(\mu-\frac{\sigma^2}{2}\Big)\Delta + \sigma \left(W_{t_{n+1}}-W_{t_n}\right) \Big\}
	\end{align*}
	and
	\begin{align*}
	\VaR_{t_n}^\alpha(\Loss_{t_n}) = & \widetilde f(t_n,X_{t_n})- e^{r\Delta} (X_{t_n}-\eta_{t_n}-\varphi_{t_n})\\
	&-\varphi_{t_n}\exp\Big\{\Big(\mu-\frac{\sigma^2}{2}\Big)\Delta+\Phi^{-1}(\alpha)\sigma\sqrt{\Delta}\Big\}
	\end{align*}
	yields that the above inequality $\VaR_{t_n}^\alpha(\Loss_{t_n}) \leq \Loss_{t_n}$ is equivalent to $\Delta^{-\frac{1}{2}}(W_{t_{n+1}}-W_{t_n}) \leq \Phi^{-1}(\alpha)$. Thus,
	\begin{align*}
	&\E\left[\Loss_{t_n}\boldsymbol{I}\left(\Loss_{t_n} \geq \VaR_{t_n}^\alpha(\Loss_{t_n})\right)\big|\cG_{t_n}\right]\\
	=&\E\Big[\Loss_{t_n}\boldsymbol{I}\Big((W_{t_{n+1}}-W_{t_n})\Delta^{-\frac{1}{2}} \leq \Phi^{-1}(\alpha)\Big)\Big|\cG_{t_n}\Big]\\
	=& \left(Y_{t_n}-e^{r\Delta}\left(X_{t_n}-\eta_{t_n}-\varphi_{t_n}\right) \right) \E\Big[\boldsymbol{I}\Big((W_{t_{n+1}}-W_{t_n})\Delta^{-\frac{1}{2}} \leq \Phi^{-1}(\alpha)\Big)\Big|\cG_{t_n}\Big]\\
	&- \varphi_{t_n}\exp\Big\{\Big(\mu-\frac{\sigma^2}{2}\Big)\Delta \Big\}\E\Big[e^{\sigma \left(W_{t_{n+1}}-W_{t_n}\right) }\\
	&\hspace*{30ex}\cdot \boldsymbol{I}\Big((W_{t_{n+1}}-W_{t_n})\Delta^{-\frac{1}{2}} \leq \Phi^{-1}(\alpha)\Big)\Big|\cG_{t_n}\Big]\\
	=&\left(\widetilde f(t_n,X_{t_n})-e^{r\Delta}\left(X_{t_n}-\eta_{t_n}-\varphi_{t_n}\right) \right)\alpha- \varphi_{t_n}\exp\Big\{\Big(\mu-\frac{\sigma^2}{2}\Big)\Delta \Big\}\\
	&\quad\cdot\E\Big[e^{\sigma \left(W_{t_{n+1}}-W_{t_n}\right) }\boldsymbol{I}\Big((W_{t_{n+1}}-W_{t_n})\Delta^{-\frac{1}{2}} \leq \Phi^{-1}(\alpha)\Big)\Big|\cG_{t_n}\Big].
	\end{align*}
	We obtain
	\begin{align*}
	&\E\Big[e^{\sigma \left(W_{t_{n+1}}-W_{t_n}\right) }\boldsymbol{I}\Big((W_{t_{n+1}}-W_{t_n})\Delta^{-\frac{1}{2}} \leq \Phi^{-1}(\alpha)\Big)\Big|\cG_{t_n}\Big]\\
	&\qquad= \int\nolimits_{-\infty}^{\Phi^{-1}(\alpha)} e^{\sigma \sqrt{\Delta}z} \frac{1}{\sqrt{2\pi}}e^{-\frac{1}{2}z^2}\d z,
	\end{align*}
	where we have used that the random variable $(W_{t_{n+1}}-W_{t_n})\Delta^{-\frac{1}{2}}$ is standard normally distributed and independent of $\cG_{t_n}$. We calculate the above integral by making an appropriate change of variables
	\begin{align*}
	\int\nolimits_{-\infty}^{\Phi^{-1}(\alpha)} e^{\sigma \sqrt{\Delta}z} \frac{1}{\sqrt{2\pi}}e^{-\frac{1}{2}z^2}\d z &= e^{\frac{\sigma^2\Delta}{2}}\int\nolimits_{-\infty}^{\Phi^{-1}(\alpha)-\sigma\sqrt{\Delta}} \frac{1}{\sqrt{2\pi}}e^{-\frac{1}{2}y^2}\d y\\
	&= e^{\frac{\sigma^2\Delta}{2}} \Phi\left(\Phi^{-1}(\alpha)-\sigma\sqrt{\Delta}\right).
	\end{align*}
	Finally, the Dynamic Tail Conditional Expectation can be written as\\ $\TCE_{t_n}^\alpha(\Loss_{t_n})=$ $\widetilde \psi(t_n,X_{t_n},\varphi_{t_n},\eta_{t_n})$, where
	\begin{align*}
	\widetilde  \psi(t,x,\bar\varphi,\bar \eta)= \widetilde f(t,x) - e^{r\Delta} (x-\bar\eta-\bar\varphi)-\frac{1}{\alpha}e^{ \mu\Delta}\Phi\left(\Phi^{-1}(\alpha)-\sigma\sqrt{\Delta}\right)\bar\varphi.
	\end{align*}
	In order to prove the statement for the Dynamic Expected Loss we use that the wealth at time $X_{t_{n+1}}$ is  distributed as $e^{r\Delta}(X_{t_n}-\eta_{t_n}-\varphi_{t_n})+\varphi_{t_n}\cdot e^Z$, where $Z$ is normally distributed with mean $m$ and variance $s^2$, where
	$
	m:= (\mu-\sigma^2/2)\Delta  \text{ and }  s^2:= \sigma^2\Delta.
	$ 
	From the definition of the Dynamic Expected Loss we find
	\begin{align*}
	\EL_{t_n}\left(\Loss_{t_n}\right) &=\E\Big[\Big(Y_{t_n} -X_{t_{n+1}}\Big)^+ \Big|\,\, \cG_{t_n}\Big]\\
	&=\E\Big[\Big(\widetilde f(t_n,X_{t_n})-e^{r\Delta}(X_{t_n}-\eta_{t_n}-\varphi_{t_n})-\varphi_{t_n}\cdot e^Z\Big)^+\Big|\,\, \cG_{t_n} \Big].
	\end{align*}
	This expectation can be calculated as follows. Let
	\begin{align*}
	f_Z(z) = \frac{1}{\sqrt{2\pi}s}\exp{\Big(-\frac{(z-m)^2}{2s^2}\Big)}
	\end{align*}
	denote the probability density function of $Z$, then	
	\begin{align*}
	\EL_{t_n}\left(\Loss_{t_n}\right) &= \int\nolimits_{-\infty}^\infty {\left(\widetilde f(t_n,X_{t_n})-e^{r\Delta}(X_{t_n}-\eta_{t_n}-\varphi_{t_n})-\varphi_{t_n}\cdot e^z\right)^+ f_Z(z)\d z}\\
	&=\int\nolimits_{-\infty}^{\widetilde  d} {\left(\widetilde f(t_n,X_{t_n}) -e^{r\Delta}(X_{t_n}-\eta_{t_n}-\varphi_{t_n})-\varphi_{t_n}\cdot e^z\right) f_Z(z)\d z}\\
	&=\left(\widetilde f(t_n,X_{t_n})-e^{r\Delta}(X_{t_n}-\eta_{t_n}-\varphi_{t_n})\right) I_1 -\varphi_{t_n}I_2,
	\end{align*}
	where
	$
	I_1=\int\nolimits_{-\infty}^{\widetilde  d}f_Z(z)\d z \text{, } I_2= \int\nolimits_{-\infty}^{\widetilde  d} e^z f_Z(z)\d z
	$ 
	and
	\begin{align*}
	\widetilde  d = \ln\Big(\frac{\widetilde f(t_n,X_{t_n})-e^{r\Delta}(X_{t_n}-\eta_{t_n}-\varphi_{t_n})}{\varphi_{t_n}}\Big).
	\end{align*}
	The integral $I_1$ can be calculated by making an appropriate change of variables and we obtain
	\begin{align*}
	I_1=\int\nolimits_{-\infty}^{d_1} { \frac{1}{\sqrt{2\pi}}e^{-\frac{y^2}{2}}\d y} = \Phi(d_1),
	\end{align*}
	where
	\begin{align*}
	d_1:&=\frac{1}{s}\Big(\ln\Big(\frac{\widetilde f(t_n,X_{t_n})-e^{r\Delta}(X_{t_n}-\eta_{t_n}-\varphi_{t_n})}{\varphi_{t_n}}\Big)-m\Big)\\
	&=\frac{1}{\sigma\sqrt{\Delta}}\Big[ \ln\Big(\frac{\widetilde f(t_n,X_{t_n})- e^{r\Delta} (X_{t_n}-\eta_{t_n}-\varphi_{t_n})}{\varphi_{t_n}}\Big)-\Big(\mu-\frac{\sigma^2}{2}\Big)\Delta \Big].
	\end{align*}
	The integral $I_2$ can be written as
	\begin{align*}
	I_2 &= \int\nolimits_{-\infty}^{\widetilde  d} {\frac{1}{\sqrt{2\pi}s}\exp{\Big(z-\frac{z^2-2zm+m^2}{2s^2}\Big)}\d z}\\
	&=e^{\frac{s^2}{2}+m}\int\nolimits_{-\infty}^{\widetilde  d} {\frac{1}{\sqrt{2\pi}s}\exp{\Big(-\frac{\left(z-(m+s^2)\right)^2}{2s^2}\Big)}\d z}.		\end{align*}
	Using the change of variables method yields
	\begin{align*}
	I_2 &=e^{\frac{s^2}{2}+m}\int\nolimits_{-\infty}^{d_2} {\frac{1}{\sqrt{2\pi}}e^{-\frac{y^2}{2}}\d y}=e^{\frac{s^2}{2}+m}\Phi(d_2),
	\end{align*}
	where
	\begin{align*}
	d_2&=\frac{1}{s}\bigg[\ln\Big(\frac{\widetilde f(t_n,X_{t_n})-e^{r\Delta}(X_{t_n}-\eta_{t_n}-\varphi_{t_n})}{\varphi_{t_n}}\Big)-(m+s^2) \bigg]\\
	&=\frac{1}{\sigma\sqrt{\Delta}}\bigg[ \ln\Big(\frac{\widetilde f(t_n,X_{t_n})- e^{r\Delta} (X_{t_n}-\eta_{t_n}-\varphi_{t_n})}{\varphi_{t_n}}\Big)-\Big(\mu+\frac{\sigma^2}{2}\Big)\Delta \bigg].		
	\end{align*}
	Note, $e^{s^2/2+m}= e^{\mu\Delta}$, thus we obtain $\EL_{t_n}\left(\Loss_{t_n}\right)= \widetilde \psi(t_n,X_{t_n},\varphi_{t_n},\eta_{t_n})$, where
	\begin{align*}
	\widetilde  \psi(t,x,\bar\varphi,\bar \eta) = \left(\widetilde f(t,x) - e^{r\Delta} (x-\bar\eta-\bar\varphi)\right)\Phi(d_1)-e^{ \mu\Delta}\Phi(d_2)\bar\varphi.
	\end{align*}
\end{appendix}

\end{document}